\documentclass[nofootinbib,prx,twocolumn,showpacs,superscriptaddress,notitlepage,amsmath,amssymb]{revtex4-2}

\usepackage[colorlinks=true,citecolor=blue,linkcolor=blue,urlcolor=blue]{hyperref}
\usepackage{bbm,tikz}

\newcommand{\FDLU}[1]{\textcolor{white}{quasi-}FDLU#1\textcolor{white}{$\qquad$}}
\newcommand{\qFDLU}[1]{quasi-FDLU#1\textcolor{white}{$\qquad$}}

\usepackage{amsmath,amssymb,bm,mathtools,amsthm}

\theoremstyle{definition}

\def\KW{\text{KW}}
\def\KWd{\widehat{\text{KW}}}
\def\Rep{\text{Rep}}
\def\SPT{\text{SPT}}

\def\aut#1#2{\sigma^{#1} \left [ {#2} \right]}

\usetikzlibrary{arrows,calc,external,shapes.geometric}
\usetikzlibrary{decorations.markings}
\tikzset{
	>=stealth',
	help lines/.style={dashed, thick},
	important line/.style={thick},
	connection/.style={thick, dotted},
}

\def\bra#1{\mathinner{\langle{#1}|}}
\def\ket#1{\mathinner{|{#1}\rangle}}
\def\bs#1{\boldsymbol{#1}}

\def\ZZ{\mathbb Z}
\def\CC{\mathbb C}
\DeclareMathOperator{\Tr}{Tr}

\makeatletter 
    
\renewcommand\onecolumngrid{
\do@columngrid{one}{\@ne}%
\def\set@footnotewidth{\onecolumngrid}
\def\footnoterule{\kern-6pt\hrule width 1.5in\kern6pt}%
}

\makeatother

\makeatletter
\def\l@subsection#1#2{}
\def\l@subsubsection#1#2{}
\makeatother

\usepackage{framed}
\definecolor{shadecolor}{gray}{0.95}
\usepackage[english]{babel}
\newtheorem{theorem}{Theorem}
\newtheorem{definition}[theorem]{Definition}
\newtheorem{cor}[theorem]{Corollary}
\newtheorem{conjecture}[theorem]{Conjecture}
\newtheorem{lemma}[theorem]{Lemma}

\begin{document}

\title{Hierarchy of topological order\\
from finite-depth unitaries, measurement and feedforward}

\author{Nathanan Tantivasadakarn}
\affiliation{Walter Burke Institute for Theoretical Physics and Department of Physics, California Institute of Technology, Pasadena, CA 91125, USA}
\affiliation{Department of Physics, Harvard University, Cambridge, MA 02138, USA}

\author{Ashvin Vishwanath}
\affiliation{Department of Physics, Harvard University, Cambridge, MA 02138, USA}

\author{Ruben Verresen}
\affiliation{Department of Physics, Harvard University, Cambridge, MA 02138, USA}

\date{\today}

\begin{abstract}
Long-range entanglement---the backbone of topologically ordered states---cannot be created in finite time using local unitary circuits, or equivalently, adiabatic state preparation. Recently it has come to light that single-site measurements provide a loophole, allowing for finite-time state preparation in certain cases. Here we show how this observation imposes a complexity hierarchy on long-range entangled states based on the minimal number of measurement layers required to create the state, which we call ``shots''. First, similar to Abelian stabilizer states, we construct single-shot protocols for creating any non-Abelian quantum double of a group with nilpotency class two (such as $D_4$ or $Q_8$). We show that after the measurement, the wavefunction always collapses into the desired non-Abelian topological order, conditional on recording the measurement outcome. Moreover, the clean quantum double ground state can be deterministically prepared via feedforward---gates which depend on the measurement outcomes. Second, we provide the first constructive proof that a finite number of shots can implement the Kramers-Wannier duality transformation (i.e., the gauging map) for any solvable symmetry group. As a special case, this gives an explicit protocol to prepare twisted quantum double for all solvable groups. Third, we argue that certain topological orders, such as non-solvable quantum doubles or Fibonacci anyons, define non-trivial phases of matter under the equivalence class of finite-depth unitaries and measurement, which cannot be prepared by any finite number of shots. Moreover, we explore the consequences of allowing gates to have exponentially small tails, which enables, for example, the preparation of any Abelian anyon theory, including chiral ones. This hierarchy paints a new picture of the landscape of long-range entangled states, with practical implications for quantum simulators.
\end{abstract}

\maketitle

\tableofcontents

\section{Introduction}
A fundamental notion emerging from decades of research into the ground states of many-body quantum systems is that of \emph{long-range entanglement} (LRE) \cite{BravyiHastingsVerstraete06,Hastings2010,ChenGuWen11A,ChenGuWen11B,ZengWen15,HuangChen2015,Haah2016}. A thermodynamically large quantum state is said to exhibit LRE if it \emph{cannot} be obtained by applying a finite-depth local unitary (FDLU) to a product state, which can intuitively be envisioned as a `brick layer' of local gates. Sometimes, one allows the gates in the circuit to have exponentially decaying tails (we refer to this as a quasi-FDLU), which are the unitary transformations generated by time evolving with a local Hamiltonian.
States related by such quasi-FDLU circuits, at least in the absence of symmetries, closely parallel the notion of a single phase of matter. Hence, LRE states represent distinct phases and cannot be obtained by adiabatic state preparation\footnote{However, there exists a small subclass, namely invertible LRE (such as the Kitaev chain) which can efficiently be prepared by simply preparing a double copy as we discuss later.}.
This is an unfortunate situation for the burgeoning field of quantum simulators where the circuit depth is a costly resource \cite{NISQ}, since the most interesting and powerful states (e.g., for quantum computation purposes) are exactly those with LRE.

However, in addition to applying quantum gates, quantum simulators and computers can perform site-resolved \emph{measurements}. In fact, it is known that measurements allow certain LRE states to be efficiently prepared in finite time, independent of the system size \cite{Briegel01,Raussendorf05,Aguado08,Bolt16,Piroli21,HastingsHaah21,measureSPT,Rydberg,Bravyi22,Lu2022,oneshot}. More generally, measurements have been known to reduce complexity of certain computational problems and its precise limits remain an active front of exploration \cite{Hoyer05,GottesmanChuang99,Jozsa06,Browne10,LiuGheorghiu21,Friedman22}. This suggests that it is worthwhile to consider a coarser equivalence than usual for phases of matter. Indeed, Ref.~\onlinecite{Piroli21} recently introduced an equivalence class for states obtainable using local unitaries of fixed depth and a \emph{sequential} (and thus linear depth) number of local operations and classical communication, which includes measurements. 

In contrast, we instead consider measurement as a scarce resource in this work. We do not place a hard cutoff on the depth of the quantum circuit as long as it is finite (i.e., does not scale with the system size), and instead ask what states are obtainable using $\ell$ rounds of single-site measurements interspersed with finite-depth unitaries. Note that the case $\ell>1$ typically requires feedforward\footnote{Indeed, if there is no feedforward, one could simply collapse multiple measurement rounds into one. To see this, note we consider the physically-motivated case of \emph{single-site} measurements (typically on ancilla qubits), which thus always commute (although the choice and ordering of unitary gates means they can correspond to effectively measuring non-commuting observables in a particular order). \label{footnote:singlemeasurement}}: an FDLU that can depend on the measurement outcomes. We denote $\ell$ as the number of `shots' we need to prepare the state.

Exploring the classification of quantum states with respect to such a number of shots has at least two important consequences. First, it gives a new conceptual and analytic tool to organize and understand interesting emergent properties of many-body quantum states.  For this particular problem, using measurement as a scarce resource further organizes the already rich classification of LRE into a hierarchy of states based on the amount of resources needed to prepare it.  Second, asking for the minimal number of shots is especially timely for the preparation of such states in noisy intermediate-scale quantum (NISQ) devices \cite{NISQ}. Although performing mid-circuit measurements are now possible, an unavoidable overhead is the fact that the protocols to prepare LRE states require the ability to adapt the circuit dynamically based on past measurement outcomes, a technology in active development \cite{DynamicCircuitIBM21,MidCircuitHoneywell21,MidCircuitIBM22,MidCircuitSandia22}. Thus, minimizing the number of mid-circuit measurements $\ell$ and feedforward can potentially lead to preparing states with higher fidelity.

\begin{figure}
    \centering
    \includegraphics{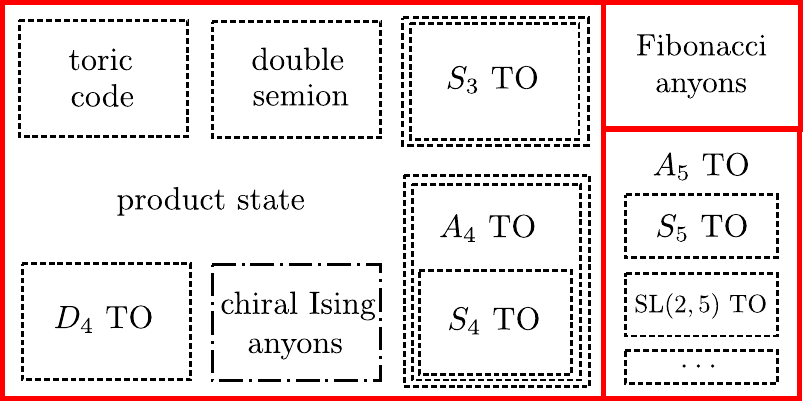}
    \caption{\textbf{Including measurements induces a hierarchy on the equivalence relation between topological orders.} Boxes are long-range entangled states which cannot be obtained from a product state by a finite-depth circuit (the topological order (TO) for a gauge group $G$ refers to the quantum double $\mathcal D(G)$). However, dashed lines indicate they can be obtained by using a layer of measurements and feedforward; the number of dashed lines equal the number of necessary shots. A dot-dashed line indicates the need for a quasi-FDLU rather than strictly local gates. Finally, solid (red) lines cannot be crossed by any finite number of measurement layers; these define non-trivial \textit{measurement-equivalent phases of matter} (red box). We argue that Fib is such an example, and similarly for non-solvable groups, with representatives for their measurement-equivalent phases being given by the quantum doubles of perfect centerless groups (e.g, $A_n$ with $n \geq 5$).}
    \label{fig:hierarchy}
\end{figure}

\begin{table*}[t]
\centering
\begin{tabular}{c||c|c}
resource & implementable states starting from product state & implementable maps\\
\hline \hline
\FDLU{} &
all cohomology SPT phases
&\FDLU{} \\ 
\qFDLU{} &
all SPT phases
& ? (all non-translation QCA?) \\ \hline
\FDLU{ + ancillas} & Kitaev chain, 3-Fermion Walker-Wang, ... & any QCA (translation, 3F QCA, ...) \\ 
\qFDLU{ + ancillas} & all invertible phases ($p+ip$, $E_8$, ...) & ? \\ \hline
$\begin{array}{c}
\textrm{\FDLU{}} \\
\textrm{+ one measurement layer}
\end{array}$ &
$\left\{ \begin{array}{l}
\textrm{all (twisted) Abelian quantum doubles (TC, DS, ...)} \\
\textrm{all nil-2 quantum doubles ($D_4$, $Q_8$, $B(2,3)$, ...)}
\end{array} \right.$
& $\left\{\begin{array}{l}
\textrm{Abelian Kramers-Wannier (KW)} \\
\textrm{Jordan-Wigner in any dimension}
\end{array}\right.$ \\
+ two measurement layers & all (twisted) metabelian quantum doubles ($S_3$, ...) & KW for all  metabelian groups \\
+ $n$ measurement layers & all (twisted) solvable quantum doubles & KW for all solvable groups \\ \hline
$\begin{array}{c}
\textrm{\qFDLU{}} \\
\textrm{+ one measurement layer}
\end{array}$  &
$\left\{ \begin{array}{l}
\textrm{all Abelian anyon theories (Laughlin $\nu=1/n$, ...)} \\
\textrm{Kitaev's 16-fold way (chiral Ising anyons, ...)}
\end{array} \right.$
& ?
\end{tabular}
\caption{\textbf{Hierarchy of quantum states via unitaries and measurement.} Implementable states and maps are stated according to the allowed resources. With only FDLUs, one can only prepare SRE phases, while invertible phases can be prepared if ancillas are allowed. With measurements, the topological order that one can prepare depends on how many measurement layers are allowed, while some topological orders are not even preparable with a finite layer of measurements, and are not included in this table (see Fig.~\ref{fig:hierarchy}).} \label{table:hierarchy}
\end{table*}

Our work reveals a new hierarchy on quantum states, of which we give an example in Fig.~\ref{fig:hierarchy}. In the largest red box, the number of dotted lines surrounding a particular state denotes the number of shots required in order to create the state from a product state. For example, the toric code and double semion, two of the simplest examples of Abelian topological order, can be prepared in one shot, while the $S_3$ quantum double, an example of a non-Abelian topological order, requires two shots. A natural conjecture is that to move from Abelian to non-Abelian topological orders, multi-shot protocols are essential. Surprisingly, this is not the case. A family of non-Abelian topological orders can be created with just a (carefully designed) single shot protocol, like their Abelian counterparts. A recent example was given for the quantum doubles of $D_4$ (the dihedral group of eight elements), with realistic gates amenable to quantum processors \cite{oneshot}. In this work, we provide a generalized protocol to prepare quantum doubles for any group of nilpotence class 2, which includes, for example, $D_4$ and $Q_8$ (the quaternion group), but not $S_3$ (the permutation group on a set of three elements). 

On the other hand, we argue that there are certain phases of matter that are not obtainable by a finite number of shots. We substantiate this for the quantum double of non-solvable groups, and also the Fibonacci topological order. In addition to being unreachable from a product state, certain non-solvable group quantum doubles are also unreachable from other other non-solvable groups. This motivates us to introduce the notion of a \textit{measurement-equivalent phase}, where in this coarser definition, two states are in the same phase (or equivalence class) if they are related to each other by a finite number of shots. Restricting to quantum doubles of finite groups, we propose a classification of such measurement-equivalent phases, shown as red boxes in Fig.~\ref{fig:hierarchy}. Each measurement-equivalent phase can be labeled by the quantum double of a perfect centerless group (with the solvable case corresponding to the trivial group). Moreover, the Fibonacci anyon theory also defines a non-trivial measurement-equivalent phase.

In addition to the hierarchy on \emph{states} (in a many-body Hilbert space), we also present a hierarchy on \emph{maps} (i.e., linear functions between many-body Hilbert spaces). Indeed, similarly to LRE states, there are certain maps of interest that cannot be written as an FDLU. Two celebrated examples are the Kramers-Wannier (KW) \cite{Wegner1971,Kogut1979,CobaneraOrtizNussinov2011,Haegeman_2015,Aasen16,VijayHaahFu2016,Williamson2016,KubicaYoshida2018,Pretko2018,ShirleySlagleChen2019,Radicevic2019} and Jordan-Wigner (JW) \cite{Jordan1928,SMLising,ChenKapustinRadicevic2018,ChenKapustin2019,Chen2019,Tantivasadakarn20,Shirley20,Po21,LiPo21} transformations. In particular we show that the Kramers-Wannier duality transformation of any solvable group $G$ can be constructed from an FDLU circuit with a finite number of shots. Here, the number of shots required is given by the derived length $l_G$ of the group, a quantity which measures how far the group is from being Abelian. This generalizes a previous result which was the Abelian case with $l_G=1$ \cite{measureSPT}.

Notably, the KW map can act on any $G$-symmetric state, which need not be a fixed point state. In fact, the input state can be critical or even long-range entangled. As a special case, choosing the input state to be the symmetric product state $\ket{+}^G_V$ gives an explicit scheme to construct any solvable quantum double. To the best of our knowledge, this is the first general protocol for such states; indeed, Ref.~\onlinecite{measureSPT} gave a general non-constructive existence argument, and Ref.~\onlinecite{Bravyi22} gave an explicit construction for the special case where each extension is split\footnote{this excludes, for instance, the quaternion group $Q_8$, though see \hyperlink{Note}{\emph{Note Added}}}. An important consequence of having an explicit KW map is that this automatically gives a way to prepare twisted quantum doubles for any solvable group (some of which are not quantum doubles of any group). Namely, we can first prepare a Symmetry-Protected Topological (SPT) state for any finite group $G$ using the FDLU from group cohomology given in Ref.~\onlinecite{Chen_2011} before applying the KW map.

\subsection{Terminology}

Let us briefly disambiguate the term \textit{feedforward} used in the paper. We follow the definition that feedforward refers to the fact that the measurement is performed in one subsystem, and the adaptive circuit is performed in a different subsystem~\cite{PhysRevApplied.9.034011}. This contrasts feedback, where the measurement and adaptive circuit act on the same subsystem. Feedforward is prominent in quantum protocols such as quantum teleportation~\cite{Teleportation,Ma2012teleportation}. Indeed, the implementation of the KW duality using measurement in~\cite{measureSPT} is akin to teleportation; after applying a quantum circuit, measurements are performed on the input subsystem, and the resulting state on the output subsystem (up to gates depending on the measurement outcome) is the KW dual of the input state.

We note that the distinction between feedforward and feedback is not necessarily a fundamental one: if a measurement of an ancilla qubit is preceded by a unitary gate, one can equally well consider the \emph{combined} object as a multi-body measurement, in which case a subsequent adaptive circuit could be seen as an example of feedback. However, we prefer to emphasize the fact that we always perform single-site measurements, for two reasons: (i) most measurement capabilities in quantum devices can indeed only measure single qubits, and (ii) it makes meaningful the notion of having a single (or multiple) layers of measurement, since single-site measurement always commute (see also footnote \ref{footnote:singlemeasurement}).

Second, we clarify the possible scenarios one can perform after measurement
\begin{enumerate}
\item One discards the measurement outcome
\item One records the measurement outcome, but one does not explicitly use it to correct the state.
\item One uses the recorded measurement outcome to act on the state
\end{enumerate}

Scenario 1 gives rise to a mixed state, whereas we wish to focus on pure states with long-range entanglement and/or topological order. Scenario 3 corresponds to a feedforward correction, which deterministically prepares the desired state. On the other hand, we do not call scenario 2 feedforward because the measurement outcome was not used to correct the state. Nevertheless, for the protocols we present which requires only one round of measurement, topological order can be prepared regardless of the measurement outcome\footnote{In this work, we consider the case where all measurement outcomes give rise to wavefunctions in the same phase of matter. In future work one can consider probabilistic versions, where, e.g., there is a finite probability of ending up in the desired phase.}, as long as the results are not discarded. More precisely, the statement is that for any measurement outcome, the resulting pure states will all be in the same non-trivial topological phase. We note that in the case of feedforward (scenario 3), we \emph{deterministically} prepare the clean state; it would be very interesting in future work to explore the \emph{probabilistic} case.

\subsection{Outline}

The sections of this paper are also structured according to this very hierarchy, ordered by the hardness of their preparation, and is summarized in Table~\ref{table:hierarchy}. In Sec.~\ref{sec:FDLU}, for completeness, we briefly review states that can already be prepared from a product state without the need of measurements. These include not only short-range entangled states, such as SPT phases, but also certain long-range entangled states, such as the Kitaev chain, provided that we use ancillas as resources. In Sec.~\ref{sec:oneshot}, we discuss states that only require one shot to prepare. This includes (twisted) Abelian quantum doubles, but also remarkably certain non-Abelian topological states, and we give an explicit protocol to construct all quantum doubles corresponding to a group of nilpotence class two. In Sec.~\ref{sec:finiteshots}, we give an explicit protocol to implement the Kramers-Wannier map for all solvable groups in a finite number of shots. This gives a method to prepare all (twisted) quantum doubles based on solvable groups. Highest in the hierarchy, in Sec.~\ref{sec:nonsolvable} we argue that there exists states that cannot be prepared by FDLU and a finite number of shots, namely the Fibonacci topological order, and the quantum doubles for non-solvable groups. Assuming this, we are able to derive how all quantum doubles of finite groups are organized into measurement-equivalent phases according to an associated perfect centerless group.  In Sec.~\ref{sec:qFDLU} we expand the allowed local unitary evolution to also include quasi-local ones (quasi-FDLU). We then reiterate through the hierarchy, showing that all invertible states can be prepared without measurement, and that certain chiral states, such as the chiral Ising topological order, can be prepared in one shot. We conclude in Sec.\ref{sec:outlook} with open questions.

\section{States obtainable without measurement}\label{sec:FDLU}

Since measurements play a key role in the results of this paper, it is equally important to review what states can already be prepared \textit{without} measurements. This is to establish that measurement is a \textit{necessary} ingredient for the scalable preparation of states in the later sections of this paper. Furthermore, these states will also serve as starting points upon which measuring gives rise to interesting states.

\subsection{FDLU: SPT states}
The first layer in the hierarchy (see Table~\ref{table:hierarchy}) naturally consists of states obtained by FDLU. In the landscape of phases of matter, this can prepare SPT states \cite{Gu09,Pollmann10,Fidkowski_2011,Turner11,Schuch11,ChenGuWen11A,ChenGuWen11B,Chen_2011,Chen_2013,pollmann_symmetry_2012,LuVishwanath12,SenthilLevin13,Levin_2012,VishwanathSenthil2013,Else_2014}, i.e., states that can only be prepared with FDLUs if the individual gates violate certain `protecting' symmetries. These states are of interest due to their entanglement structure, which in the case of cluster or graph states (obtained by applying a circuit of Controlled-$Z$ gates) can be used, e.g., for measurement-based quantum computation \cite{GottesmanChuang1999,BriegelRaussendorf2001,RaussendorfBriegel2001,RaussendorfBrowneBriegel2003}. In fact, we will discuss how the interesting short-range entanglement of various SPT states can be used to construct LRE via measurement \cite{measureSPT}.

\subsection{FDLU + ancillas: invertible states and QCAs} The next step in the hierarchy does not yet involve measurement, but merely ancilla qubits. Remarkably, there exist states which can \emph{only} be created from FDLU if one uses such ancillas. These states are invertible LRE states, such as the Kitaev-Majorana chain, i.e., the $p$-wave superconducting chain. Indeed, an FDLU \emph{can} create \emph{two decoupled} Kitaev chains (see Appendix \ref{app:Kitaev}). If we simply remove a single copy\footnote{This is usually forbidden in the ``stabilization" of a phase by ancillas: only product states can be removed. Nevertheless, see Ref.~\onlinecite{ShirleySlagleChen19_1} for a coarser definition of phase in fracton orders.}, we have thus obtained the Kitaev-Majorana chain.

More generally, one can ask about the class of maps obtained from FDLU and ancillas. This turns out to contain all locality preserving unitaries, also called Quantum Cellular Automata (QCA). Indeed, it is known that if a QCA acts on a Hilbert space $\mathcal H$, then there exists an FDLU for QCA $\otimes$ QCA$^{-1}$ on the doubled Hilbert space $\mathcal H \otimes \mathcal H$ \cite{HaahFidkowskiHastings18}. For instance, this allows the implementation of the translation operator on $\mathcal H$, which is yet another way to obtain the Kitaev-Majorana chain. Other interesting QCAs in higher dimensions have emerged in the past few years \cite{HaahFidkowskiHastings18,Haah21,Shirley22,Haah22}.

Similarly, while SPT states cannot be prepared using a finite depth of local gates that preserve the symmetry, they can still be prepared with the help of ancillas by first using a symmetric circuit to prepare the state SPT $\otimes$ SPT$^{-1}$ and then removing a single copy.

\section{States preparable in one shot: Abelian and nil-2 non-Abelian quantum doubles}\label{sec:oneshot}
It is known that certain (non-invertible) LRE states, including the Greenberger-Horne-Zeilinger (GHZ) state, toric code, and in fact any Calderbank-Shor-Steane (CSS) code, can be prepared by measuring cluster states \cite{Briegel01,Raussendorf05,Aguado08,Bolt16,Piroli21}. Recently, the present authors, in collaboration with Ryan Thorngren, generalized this by showing how FDLU and a single measurement layer can \emph{implement the Kramers-Wannier (KW) transformation for any Abelian symmetry} in any spatial dimension \cite{measureSPT}. This can be thought of as a protocol to `gauge' an Abelian symmetry. We will briefly recap this in Sec. \ref{sec:KWabelian} which allows us to prepare any twisted Abelian quantum double, as already noted in Ref.~\onlinecite{measureSPT}. We then point out in Sec.~\ref{sec:nil2} that beyond obtaining Abelian topological order, a single-shot protocol can even prepare a class of non-Abelian topological orders corresponding to the quantum doubles of nilpotent groups of class two, such as $D_4$ or $Q_8$. 

\subsection{Twisted Abelian quantum doubles}\label{sec:KWabelian}

Let us briefly recap this `gauging' protocol (or Kramers-Wannier (KW) duality) for a global $\mathbb Z_2$ symmetry generated by $\prod_v X_v$ (we denote the Pauli matrices by $X,Y,Z$) on an arbitrary cellulation of a closed spatial manifold. Starting with an \emph{arbitrary} $\ZZ_2$ symmetric wavefunction $|\psi\rangle_V$ defined on the vertices of a lattice, we introduce product state ancillas $\ket{\uparrow}_E$ on the edges. After applying Controlled-Not gates (denoted $CNOT$ or $CX$ more briefly) to all nearest-neighbor bonds (with control on vertices and target on edges), projecting the vertices into a symmetric product state implements the KW map: 
\begin{equation}
\KW_{EV}^{\mathbb Z_2} |\psi\rangle_V = \langle +|_V \prod_{\langle v,e\rangle} CX_{ve} |\uparrow\rangle_{E} \otimes |\psi\rangle_V. \label{eq:KW}
\end{equation}
Here, we use the shorthand $\ket{\uparrow}_E = \bigotimes_{e\in E}\ket{\uparrow}_e$ and  $\ket{+}_V = \bigotimes_{v\in V}\ket{+}_V$.  Crucially, Eq.~\eqref{eq:KW} does \emph{not require post-selection}: if one measures the vertices in the $X$-basis and finds $|-\rangle$, these can always be paired up\footnote{Symmetry dictates these always come in pairs.} in one extra unitary layer using string operators. In particular, a pair of $\ket{-}$ on two vertices can be turned into a $\ket{+}$ by applying a string of $Z_e$ on all edges connecting these two vertices. This is the only step relying on the symmetry group being \emph{Abelian}: the measurement outcomes label Abelian gauge charges (anyons in 2D) which can always be paired up in finite time. Applying Eq.~\eqref{eq:KW} to the special case of a 2D product state deterministically prepares the toric code, whereas choosing a $\mathbb Z_2$ SPT phase \cite{Levin_2012} (which is itself preparable using FDLU, since the preparing circuit does not need to respect symmetry.) leads to double-semion (DS) topological order \cite{measureSPT}.

Let us also note two related versions of the KW map: (i) the $\KW$ in Eq.~\eqref{eq:KW} maps the Ising interaction as $Z_v Z_{v'} \to Z_e$. Alternately, (ii) applying the Hadamard gate on all edges replaces $CX$ by $CZ$ (i.e., Controlled-Z) and $\ket{\uparrow}_E$ by $\ket{+}_E$ giving the map $Z_v Z_{v'} \to X_e$, which we denote by $\KWd$. Specifically,
\begin{align}
    \KWd_{EV}^{\mathbb Z_2} &= \left(\prod_e H_e \right) \KW _{EV}^{\mathbb Z_2} \left( \prod_e H_e\right) \\ &= \bra{+}_V\prod_{\langle v,e\rangle} CZ_{ve} \ket{+}_{E}.
\end{align}
in which we recover the protocol of measuring cluster states by using $\ket{+}_V$ as an input.

The measurement protocol to implement the KW duality generalizes naturally to any finite Abelian group $A$, by using the corresponding generalizations of $CX$ and $CZ$. In this case, the graph must now be \textit{directed}. Each edge $e$ can be associated with an ``initial" vertex $i_e$ and ``final" vertex $f_e$, and reciprocally, for each vertex $v$ we denote the set of edges pointing into and out of $v$ as $e\rightarrow v$ and $e\leftarrow v$, respectively. The KW map is constructed as
\begin{align}
    \KW_{EV}^{A} &=  \bra{+}_V \prod_v \left [\prod_{e\rightarrow v} (CX^A_{ve})^\dagger \prod_{e\rightarrow v} CX^A_{ve}\right] \ket{1}_{E}\\
    &= \bra{+}_V \prod_{e} CX^A_{i_ee}  (CX^A_{f_ee})^\dagger \ket{1}_{E}
    \label{eq:KWabelian}
\end{align}
where $\ket{+} = \frac{1}{|A|} \sum_{a \in A} \ket{a}$, $1$ is the identity in $A$, and the generalized $CX$ gate for the group $A$ is defined as
\begin{align}
    CX^A_{ve} \ket{a_v,a_e} = \ket{a_v,a_va_e} 
\end{align}
Similarly, $ \KWd_{EV}$ can be obtained by performing the Fourier transform on all edges, which changes $CX$ to $CZ$ and $\ket{1}$ to $\ket{+}$. We get
\begin{align}
\KWd_{EV}^{A} &=  \bra{+}_V \prod_v \left [\prod_{e\rightarrow v} (CZ^A_{ve})^\dagger \prod_{e\rightarrow v} CZ^A_{ve}\right] \ket{+}_{E}\\
    &= \bra{+}_V \prod_{e} CZ^A_{i_ee}  (CZ^A_{f_ee})^\dagger \ket{+}_{E}
    \label{eq:KWdabelian}
\end{align}
where the generalized $CZ$ gate is defined as
\begin{align}
    CZ^A_{ve}\ket{a_v,a_e} &= \chi^{a_v}(a_e)\ket{a_v,a_e}.
\end{align}
Here, we use the fact that for Abelian groups, there is an isomorphism between group elements and irreps so that we can define $\chi^{a}$, the character corresponding to the group element $a$. Again, post-selection is not required since one can always pair up the corresponding charges in finite time.

Nevertheless, the above construction does not naively apply to the KW map for non-Abelian groups (which would be a method to prepare non-Abelian topological order), since FDLU cannot pair up the non-Abelian anyons that result as measurement outcomes\cite{Shi19} (see Sec.~\ref{sec:KWG} for a full discussion on this obstruction). However, some groups are obtained by a finite number of Abelian extensions. E.g., the symmetry group of the square, $D_4$, can be obtained by extending $\mathbb Z_2 \times \mathbb Z_2$ (the horizontal and vertical mirror symmetries) by $\mathbb Z_2$ (the diagonal mirror); note that these two do not commute\footnote{Indeed, the product of the horizontal and diagonal mirror symmetries gives the $\mathbb Z_4$ rotation.}. Using this observation, Ref.~\onlinecite{measureSPT} observed that successively applying Eq.~\eqref{eq:KW} can thus generate any quantum double for a \emph{solvable} gauge group; Ref.~\onlinecite{Rydberg} presented explicit two-step protocols for $D_4$ and $S_3$ (see also Ref.~\onlinecite{Bravyi22}).

The above line of reasoning strongly suggests that it is impossible to create non-Abelian topological order in a \emph{single} shot. Perhaps surprisingly, this expectation is false. In Ref.~\onlinecite{oneshot}, the present authors argued that it is in fact possible to prepare non-Abelian topological order that admits a Lagrangian subgroup in a single shot and gave explicit protocols to prepare the $D_4$ and $Q_8$ topological orders. In this work, we present an explicit protocol to prepare a class of non-Abelian topological orders, all of which can be prepared in a single shot: the quantum double for class-2 nilpotent groups.

\begin{figure*}[t!]
    \centering
    \includegraphics[scale=0.9]{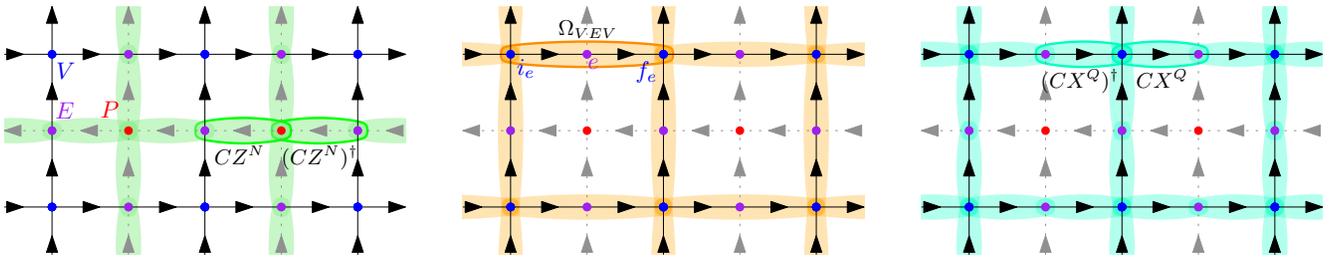}
    \caption{\textbf{Single-shot preparation of the quantum double for nil-2 groups.} The three circuit layers needed to entangle the product state in Eq.~\eqref{eq:nil2}. After measuring vertices (blue) and plaquettes (red), the resulting state on the edges (purple) exhibits $G$ topological order regardless of measurement outcome. If one desires, the exact ground state of $\mathcal D(G)$ can be recovered by a string of $X^n$ on the dual lattice and a string of $Z^q$ on the direct lattice. }
    \label{fig:nilpotent}
\end{figure*}

\subsection{Quantum double for class-2 nilpotent $G$ in one shot}\label{sec:nil2}
A group $G$ is class-2 nilpotent (commonly called a nil-2 group) if there exists finite Abelian groups $N$ and $Q$ such that the extension
\begin{equation}
 1 \rightarrow N  \xrightarrow[]{} G  \xrightarrow[]{} Q \rightarrow 1
 \label{eq:centralextension}
\end{equation}
is \textit{central}. That is, $N$ is contained in the center $Z(G)$. Such central extensions are specified by a function $\omega:Q^2\rightarrow N$ called a 2-cocycle $(\omega \in H^2(Q,N))$, which determines how multiplication of $Q$ can give rise to elements in $N$. As a 2-cocycle, $\omega$ satisfies the cocycle condition
\begin{align}
    \omega(q_2,q_3) \omega(q_1,q_2q_3)  = \omega(q_1,q_2) \omega(q_1q_2,q_3).
\end{align}
Elements in $g\in G$ can be denoted by the pair $(n,q)$ whose group law is given via
\begin{align}
    (n_1,q_1) \times (n_2,q_2) = (n_1n_2\omega(q_1,q_2),q_1q_2).
\end{align}

As an example, consider $N=\ZZ_2$ and $Q=\ZZ_2^2$. Denote an element $q\in Q$ as the pair $(a,b)$ where $a,b \in \{0,1\}$ with addition as the group multiplication. One choice of a cocycle is
\begin{align}
    \omega((a_1,b_1),(a_2,b_2)) = a_1b_2
\end{align}
One can check that the above cocycle condition is satisfied, and by further checking the group multiplication, one finds the resulting group is the dihedral group $G=D_4$. On the other hand, the cocycle
\begin{align}
    \omega((a_1,b_1),(a_2,b_2)) = a_1a_2 +a_1b_2 + b_1b_2
\end{align}
gives rise to the quaternion group $G=Q_8$.

In previous proposals \cite{measureSPT,Bravyi22} such quantum double $\mathcal D(G)$ requires two rounds of measurements by sequentially performing KW on $N$ then $Q$. Conceptually, pairing up the charges in the first round before continuing is crucial to avoid creating non-Abelian charges in the second measurement round.  In the current proposal, we can prepare the same state in one shot by instead gauging a particular $N\times Q$ SPT state. The Abelian charges we measure by gauging this SPT can be translated into a combination of Abelian charge and fluxes of $\mathcal D(G)$, which braid trivially.

We now give the exact claim for preparing the ground state of the quantum double model for a nil-2 group $G$ on the edges $E$ (purple) of square lattice using ancillas on the vertices $V$ (blue) and plaquettes $P$ (red) as in Fig. \ref{fig:nilpotent} (though it applies to arbitrary graphs).  The local Hilbert space is given by the group algebra with basis elements $\ket{q_v} \in \CC[Q]$ on $V$, $\ket{n_p} \in \CC[N]$ on $P$, and  $\ket{n_e,q_e} \in \CC[N]\times \CC[Q] \cong \CC[G]$ on $E$. Each edge is given two directions: one connects between vertices (black arrows) and one connects plaquettes (grey arrows).

\begin{lemma}
   \begin{shaded} 
   The ground state for the nil-2 quantum double can be expressed as
   \begin{align}
    \ket{\mathcal D(G)}_E &=\bra{+}_{VP} CX^Q_{VE}   \Omega_{VEV}    CZ^N_{PE} \ket{+}_{V}\ket{+,1}_E\ket{+}_P \nonumber\\
      &=\KW^Q_{EV} \Omega_{VEV} \KWd^N_{EP} \ket{+}_V \ket{+}_P \label{eq:nil2}
\end{align}
where the action of $\Omega_{VEV}$ is defined as\footnotemark
\begin{align}
      \Omega_{VEV} \ket{ \{q_v\},\{n_e\} } &= \ket{ \{q_v\},\{n_e \bar \omega(q_{i_e},\bar q_{i_e}q_{f_e})\} } \label{eq:OmegaVEV}
\end{align}
   \end{shaded}
\end{lemma}
\footnotetext[7]{Again, bars denote group inverses.}

We defer the proof that the resulting state is indeed exactly the ground state of $\mathcal D(G)$ to Appendix \ref{app:nil2prepproof}, where we use properties of KW maps for normal subgroups, developed in Sec.~\ref{sec:KWNinG}. Using the above result, we are able to show that

\begin{theorem}
\begin{shaded}
The protocol Eq.~\eqref{eq:nil2} for preparing a nil-2 quantum double can be performed in a single shot.
\end{shaded}
\end{theorem}

\begin{proof}
 We postpone all measurements in the protocol until the very end. Since the KW maps act on different subspaces, it is possible to correct the measurement outcomes independently. Namely, measurement outcomes of the vertices and plaquettes correspond to charges $q \in Q$ and fluxes of $n\in N$, which can be paired up using solid strings of $Z_e^q$ and dotted strings of $X_e^n$, respectively.
\end{proof}
The interpretation of the protocol is as follows. The first layer (had the measurements on the plaquettes been immediately performed) gauges the symmetric product state on $N$, thus preparing the $N$-toric code. In the second step, $\Omega_{VEV}$ turns the toric code into a Symmetry-Enriched Topological (SET) state protected by $Q$. In particular, the charges of the toric code are fractionalized by the symmetry $Q$, and the fractionalization is given precisely by the cocycle $\omega \in H^2(Q,N)$\footnote{In fact, the cocycle $\omega(q_{i_e},\bar q_{i_e}q_{f_e})$ is exactly what appears in the group cohomology construction of 1+1D SPTs in Ref.~\onlinecite{Chen_2013}.}. Finally, the last layer along with the measurement on the vertices gauges $Q$. 

A couple of remarks are in order. First, regardless of measurement outcome, the state always has $\mathcal D(G)$ topological order since the feedforward correction is pairing up Abelian anyons\footnote{To be precise, the pure state corresponding to a given measurement outcomes has topological order. On the other hand, if one discards the measurement outcomes, then the resulting mixed state cannot be said to have topological order. }. Second, it is possible to view the above protocol as gauging an $N\times Q$ decorated domain wall SPT state\cite{decorateddomainwalls}, where a 1+1D $Q$-SPT state is decorated on $N$ domain walls. We elaborate on this point in Appendix~\ref{app:DDW}.

Lastly, we make contact to a sufficient condition in Ref.~\onlinecite{oneshot} that any anyon theory that admits a Lagrangian subgroup can be prepared in one shot. A Lagrangian subgroup $A$ is a subset of Abelian bosons in the topological order that are closed under fusion, have trivial mutual statistics, and that every other anyon braids non-trivially with at least one of the anyons in the subgroup\cite{KS11,Levin13}. We first recall that the pure charges and fluxes of $\mathcal D(G)$ are labeled by irreps and conjugacy classes of $G$, respectively. There are two natural classes of Abelian anyons when $G$ is nil-2. First, since $Q$ is Abelian, irreps of $Q$ are all one-dimensional, thus they pullback to Abelian charges in $G$. Second, since $N$ is in the center of $G$, conjugacy classes of $N$ remain one-dimensional upon being pushed forward to $G$, giving Abelian fluxes. Moreover, these two classes of anyons have mutual braiding statistics, which follows from the exactness of the sequence \eqref{eq:centralextension}: pulling irreps from $Q$ back to $N$ gives the trivial irrep. This can be interpreted physically as the fact that $Q$ charges are transparent to the $N$ fluxes. Together, these fluxes and charges together form a subgroup $A=N\times Q$. Lastly, a sufficient condition to verify that the subgroup is Lagrangian is that the size of the group is equal to the total quantum dimension of the theory. Indeed, one has $|A| = |G|$. To conclude, for the quantum double of a nil-2 group, conjugacy classes of $N$ and irreps of $Q$ form a Lagrangian subgroup of $\mathcal D(G)$. These are exactly the anyons we measure to prepare the state in one shot.
\begin{table*}[t!]
    \centering
    \begin{tabular}{|c|c|c|c|}
    \hline
    Group & Nil-2 & Metabelian & Solvable\\
    \hline
    $\KW^G$ & $\KW^Q \times   \Omega \times  \KW^N$ \eqref{eq:2stepnil2}    & $\KW^Q \times \mathcal U^{N \triangleleft G}\times  \KW^N $  \eqref{eq:2stepfull}
         & $\prod_{j=l_G}^{1} \KW^{{N_{j}\over N_{j-1}}\triangleleft{ G\over N_{j-1}}}$ \eqref{eq:sequentialgauging}\\
         \hline
    \end{tabular}
    \caption{Decomposition of $\KW^G$ in terms of unitaries and Abelian KW operators, the latter of which can be performed using unitaries and one round of measurement. For the solvable case, $N_j$ are subgroups in the derived series of $G$, and $\KW^{N \triangleleft G} = \mathcal U^{N \triangleleft G} \KW^N$ is the KW map to gauge a normal subgroup $N$ of $G$ (defined in Eq.~\eqref{equ:KWonstatesNG}) and $\mathcal U^{N \triangleleft G} = \Sigma^{-1} \times \Omega$, where $\Sigma$ and $\Omega$ are unitary operators that depend on the the factor system of the group extension. When $l_G=2$, the result reduces to the metabelian case. When the group extension of the metabelian group is central, $\Sigma = \mathbbm 1$ and the result reduces to the nil-2 case.}
    \label{tab:KWsummary}
\end{table*}

\begin{table*}[]
    \centering
     \begin{tabular}{|c|c|c| c|c|}
    \hline
    $G$ & $N$ & $Q$ & $\sigma$ & $\omega$ \\
    \hline
    $S_3$ & $\ZZ_3 = \{0,1,2\}$ &$\ZZ_2 = \{0,1\}$ & $\sigma^1[n]=-n$ & $\omega(q_1,q_2)=1$  \\
    \hline
    $D_4$ & $\ZZ_2 =\{0,1\} $ & $\ZZ_2^2 = \{0,1\}^{2}$ & $\sigma^q =1$ & $\omega((a_1,b_1),(a_2,b_2)) = a_1b_2$\\
    \hline
    $Q_8$ & $\ZZ_2= \{0,1\}$ & $\ZZ_2^2 = \{0,1\}^{2}$ & $\sigma^q =1$ & $\omega((a_1,b_1),(a_2,b_2)) = a_1a_2 +a_1b_2 + b_1b_2$ \\
    \hline
    \end{tabular}
    \caption{Examples of group extensions and factor systems for metabelian groups. Here, we assume addition as the group multiplication for the normal and quotient subgroups. The map $\KW^G$ for such groups can be implemented using the decomposition in Eq.~\eqref{eq:2stepfull}}
    \label{tab:metabelianexamples}
\end{table*}

\section{Finite-shots: Kramers-Wannier for solvable groups}\label{sec:finiteshots}
In previous work \cite{measureSPT}, we argued that quantum doubles corresponding to solvable groups---groups that arise from recursively extending Abelian groups---can be prepared in finite time. In the present work, we show that this can be extended to the space of maps: one can implement the Kramers-Wannier map $\KW^G$ for any solvable group $G$ in finite time using measurements and feedforward. Applying it to the symmetric product state then gives the $G$ quantum double, constituting the first explicit protocol for splitting-simple solvable groups. Moreover, by applying it to $G$-SPT states yields all twisted Abelian quantum doubles.

\subsection{Statement of the results}

In this subsection, we present the results, which are then derived and explained in the remainder of this section.
Rather than jumping straight to the result for arbitrary solvable groups, we offer two stepping stones where we introduce the ingredients necessary for the general case, as summarized in Table~\ref{tab:KWsummary}. When we write $\KW^G$ for a non-Abelian group $G$, we refer to the particular gauging map which is defined by the Kramers-Wannier transformation; we review it in more detail in Sec.~\ref{sec:KWG}.

\subsubsection{Two-shot protocol for gauging nil-2 groups}
Let us first consider a group $G$ of nilpotency class two, as in Sec.~\ref{sec:nil2}. This means that $G$ is obtained via a central extension \eqref{eq:centralextension} involving the normal subgroup $N$ and quotient group $Q$. In Sec.~\ref{sec:nil2}, we prepared the quantum double $\mathcal D(G)$ in a single shot using Eq.~\eqref{eq:nil2}. The approach used there essentially involved gauging\footnote{This is most explicit by writing Eq.~\eqref{eq:centralextension} as gauging $N \times Q$ after applying a particular SPT-entangler (see Appendix~\ref{app:nil2}); this is also called a choice of defectification class of the gauging map \cite{BarkeshliBondersonChengWang2019}.} a product group $N \times Q$ and thus cannot be used to gauge $G \ncong N \times Q$ for a generic $G$-symmetry input state.

Indeed, although we could prepare the $G$ quantum double in a single shot, here we only find a two-shot protocol for gauging $G$. One way to naturally motivate this is by first considering an alternate protocol for the aforementioned quantum double. Indeed, instead of starting with measuring the plaquette term to prepare the $N$-toric code using $\KWd^N_{PE}$ (as in Eq.~\eqref{eq:KWdabelian}), we can measure the vertex terms of the toric code using $\KW^N_{EV}$. This has two consequences. First, since the charges carry fractional charge under $Q$, they become non-Abelian anyons in the quantum double; this forces us to apply feedforward after this first gauging step (when charges are still Abelian), making it a two-shot protocol. Second, now the input state is $\ket{+}^N_V \otimes \ket{+}^Q_V \cong \ket{+}^G_V$, which can be interpreted as a $G$-symmetric state, suggesting that we can indeed interpret this approach as gauging the $G$ symmetry of the product state (which is a well-known way of producing $\mathcal D(G)$). Indeed, later in this section we will see that this protocol gauges the $G$ symmetry for any $G$-symmetric input state, i.e.:
\begin{align}
    \KW^G_{EV} = \KW^Q_{EV}  \Omega_{VEV} \KW^N_{EV}, \label{eq:2stepnil2}
\end{align}
where we remind the reader that $\KW$ for Abelian groups is defined in Eq.~\eqref{eq:KWabelian} and the unitary $\Omega_{VEV}$ in Eq.~\eqref{eq:OmegaVEV}.

\subsubsection{Two-shot protocol for gauging metabelian groups}

Next, we consider the most general type of non-Abelian groups that are obtained by a single Abelian extension of an Abelian group, called metabelian groups. In contrast to the aforementioned case of nil-2 groups (which is more restrictive), the extension \eqref{eq:centralextension} defining a metabelian group need not be central.

Hence, to characterize the extension that gives rise to $G$, one must specify---in addition to the 2-cocycle $\omega$---how $Q$ acts on $N$. This is given by a map $\sigma: Q \rightarrow \text{Aut}(N)$. That is, for a fixed $q$, $\sigma^q: N \rightarrow N$ is an automorphism.
The multiplication law is now
\begin{align}
    (n_1,q_1) \times (n_2,q_2) = (n_1\aut{q_1}{n_2}\omega(q_1,q_2),q_1q_2)
    \label{eq:multiplicationlaw}
\end{align}
and associativity demands that $\omega$ satisfies the cocycle condition
\begin{equation}
    \aut{q_1}{\omega(q_2,q_3)} \omega(q_1,q_2q_3)  = \omega(q_1,q_2) \omega(q_1q_2,q_3).
    \label{equ:cocyclecondition}
\end{equation}
The pair $\sigma$ and $\omega$ is together called a \textit{factor system} (see Appendix \ref{app:factor} for a derivation of the above properties). For instance, $S_3 \cong \mathbb Z_3 \rtimes \mathbb Z_2$ is a (split) extension with the normal subgroup $N=\mathbb Z_3 = \{0,1,2\}$ and quotient group $Q=\mathbb Z_2 = \{ 0,1\}$, with addition as the group multiplication in the normal and quotient subgroups. while the cocycle is trivial, the extension has a non-trivial automorphism $\sigma^1[n] = -n$. Examples for other metabelian groups can be found in Table~\ref{tab:metabelianexamples}.

Like the nil-2 case above, we present a two-shot protocol for gauging such a metabelian group, using the above data $\omega$ and $\sigma$. In particular, we claim that $\KW^{G}$ can be prepared by inserting a specific FDLU $\mathcal U^{N \triangleleft G}$ between the two KW maps: 
\begin{lemma}
\begin{shaded}
For an arbitrary group extension, we have the following identity:
    \begin{align}
    \KW^G_{EV} = \KW^Q_{EV} \mathcal U^{N \triangleleft G}_{EV} \KW^N_{EV} \label{eq:2stepfull}
\end{align}
where $\mathcal U^{N \triangleleft G}$ is the FDLU
\begin{align}
     \mathcal U^{N \triangleleft G}_{EV} &= \Sigma_{EV}^{-1} \Omega_{VEV}
     \label{eq:UasFDLU}
\end{align}
with $\Omega_{VEV}$ defined in Eq.~\eqref{eq:OmegaVEV} and the action of $\Sigma_{EV}$ is given by
\begin{align}
    \Sigma_{EV} \ket{\{q_v\},\{n_e\}}= \ket{\{q_v\},\{\sigma^{q_{i_e}}[n_e]\}}.
\end{align}
\end{shaded}
\label{lem:metabelian}
\end{lemma}
 We provide a physical intuition for the role of $\mathcal U^{N \triangleleft G}$ both from the point of view of a basis transformation, and as an entangler that symmetry-enriches the input state in Sec.~\ref{sec:Uphysical}. The full proof is found in Appendix~\ref{app:KWN}.

The above result implies that for a metabelian group $G$ where $Q$ and $N$ are Abelian, we can implement $\KW^G$ using two rounds of measurement. Applying  $\KW^G$ to the $G$-symmetric product state prepares $\mathcal D(G)$ for any metabelian group. For $N=\ZZ_3$ and $Q=\ZZ_2$, the protocol matches our previous proposal to prepare $\mathcal D(S_3)$ in Ref.~\onlinecite{Rydberg}, and also agrees with Ref.~\onlinecite{Bravyi22} (up to an appropriate inversion of group elements) which treated the case of a split extension (i.e., for a trivial cocycle $\omega$). We further remark that in the special case of a central extension (where $\sigma$ is trivial), we have $\Sigma_{EV} = 1$ and thus recover the nil-2 protocol in Eq.~\eqref{eq:2stepnil2}.

In fact, for the remainder of this section, we find it convenient to define $\KW^{N \triangleleft G}$ as the KW map for gauging a normal subgroup $N$ of $G$\footnote{In this notation, $\KW^{G\triangleleft G} = \KW^G$.}. Namely, it is the canonical choice for which the following property for two-step gauging holds
\begin{align}
    \KW^G = \KW^Q \KW^{N \triangleleft G}
    \label{eq:2step}
\end{align}
Intuitively, the map $\KW^{N \triangleleft G}$ gauges $N$ in such a way that it leaves the action of the quotient group $Q$ untouched, which implies that $Q$ can be sequentially gauged as is (see Eq.~\eqref{equ:KWonstatesNG} in Sec.~\ref{sec:KWNinG} for a proper definition.). Hence, the non-trivial result in Eq.~\eqref{eq:2stepfull} boils down to the statement that
\begin{align}
    \KW^{N \triangleleft G} = \mathcal U^{N \triangleleft G} \KW^N. \label{eq:KWNGdef}
\end{align}
We will use the above form for $\KW^{N \triangleleft G}$ moving forward.

\subsubsection{Finite-shot protocol for gauging solvable groups}

To state our most general result, let us first briefly review some useful notions about solvable groups. A derived series is a set of normal subgroups $\tilde N_i \triangleleft G$,  defined inductively by
\begin{equation}
\begin{split}
    \tilde N_0 &= G\\
    \tilde N_j' &= [\tilde N_{j-1},\tilde N_{j-1}] ; \ \ \ j>0
    \label{eq:derivedseries}
    \end{split}
\end{equation}
where $[\tilde N,\tilde N]$ is the commutator subgroup of $\tilde N$. The smallest natural number $l_G$ where $\tilde N_{l_G}= \ZZ_1$ is called the \textit{derived length} of the group. A solvable group is defined as a group where $l_G$ is finite. For example, Abelian groups, and metabelian groups correspond to $l_G=1$ and $l_G=2$, respectively.

To simplify the notation, we now define $N_j = \tilde N_{l_G-j}$ so that $N_0=\ZZ_1$ and $N_{l_G}=G$. We also remark that using the fact that $N_j$ are commutator subgroups, one can show that for the derived series, $N_i \triangleleft N_k$ for all $i<k$, and that $N_{j+1}/N_j$ are all Abelian.

Based on the derived series, we claim that $\KW^G$ can be implemented in exactly $l_G$ rounds of measurement, which we moreover expect to be optimal. Specifically:
\begin{theorem}
\begin{shaded}
The Kramers-Wannier map for any solvable group $G$ with derived length $l_G$ can be implemented using a finite-depth unitary and $l_G$ measurement layers (interspersed with feedforward) as:
\begin{align}
    \KW^G &=\prod_{j=l_G}^{1} \KW^{{N_{j}\over N_{j-1}}\triangleleft{ G\over N_{j-1}}} \\
    &= \KW^{G \over N_{l_G-1}}   \KW^{{N_{l_G-1} \over N_{l_G-2}}  \triangleleft {G \over N_{l_G-2}}}   \cdots   \KW^{{N_{2} \over N_{1}}  \triangleleft {G \over N_{1}}}   \KW^{N_{1} \triangleleft G}, \nonumber \label{eq:sequentialgauging}
\end{align}
where each $\KW^{N \triangleleft G}$ can be performed in a single shot using Eq.~\eqref{eq:KWNGdef}.
\end{shaded}
\label{thm:Gsolvable}
\end{theorem}

Intuitively, at round $j=1$, we begin by performing the Abelian KW for the normal subgroup $N_{1}$. This leaves a remaining quotient symmetry $G/N_{1}$. We now proceed inductively for $j=2,\ldots, l_G$. At round $j$, we have gauged the group $N_{j-1}$, so we proceed to gauge the group $N_{j}/N_{j-1}$ which is again an Abelian normal subgroup of the remaining symmetry $G/N_{j-1}$. This reduces the remaining symmetry to $G/N_{j}$. We repeat this until the entire symmetry $G$ is gauged.

\begin{proof}
We proceed by induction. The base case $l_G=1$ is trivial, and for $l_G=2$, Lemma \ref{lem:metabelian} gives the existence of a map $\KW^{N \triangleleft G}$ which satisfies the two-step gauging condition Eq.~\eqref{eq:2step}. 

Now we proceed to the induction step. Suppose that we have proven Eq.\eqref{eq:sequentialgauging} for derived length $l_G$,  consider a group $G'$ of derived length $l_{G'} = l_{G}+1$ with derived series $N'_{j}$. 
Let $G= G'/N_{1}'$, and $N_j=N_{j+1}'/N_1'$. From the third isomorphism theorem, we have that
\begin{align}
    \frac{G}{N_{j-1}} &= \frac{G'/N_{1}'}{N'_{j}/N'_{1}} = \frac{G'}{N'_{j}}, & \frac{N_j}{N_{j-1}} &=\frac{N_{j+1}'/N_{1}'}{N'_{j}/N'_{1}} = \frac{N_{j+1}'}{N'_{j}}.
\end{align}
Then,

\begin{align}
\begin{split}
  \prod_{j=l_{G'}}^{1} \KW^{{N'_{j}\over N'_{j-1}}\triangleleft{ G'\over N'_{j-1}}} &= \left[\prod_{j=l_{G}}^{1} \KW^{{N'_{j+1}\over N'_{j}}\triangleleft{ G'\over N'_{j}}}\right] \KW^{N_1' \triangleleft G'} \\
&= \left[\prod_{j=l_{G}}^{1} \KW^{{N_{j}\over N_{j-1}}\triangleleft{ G\over N_{j-1}}}\right] \KW^{N_1' \triangleleft G'}\\
&=\KW^G \KW^{N_1' \triangleleft G'}\\
& = \KW^{\frac{G'}{N_1'}} \KW^{N_1' \triangleleft G'} =\KW^{G'}  
\end{split}
\end{align}
as desired.
\end{proof}

The remaining of this section is devoted to deriving properties of $\KW^{N \triangleleft G}$ Due to its length, we provide a brief summary. Sec.~\ref{sec:KWG} reviews the KW map for an arbitrary finite group $G$, and discusses the subtleties of implementing the map with FDLU and measurements for non-Abelian groups. In Sec.~\ref{sec:KWNinG}, we give a prescription of the $\KW^{N \triangleleft G}$ and show that for Abelian $N$, this map can be implemented with FDLU and a single round of measurement and feedforward.
In. Sec.~\ref{sec:Uphysical}, we provide intuition on how and why $\KW^{N \triangleleft G}$ differs from $\KW^{N}$ from the viewpoint of symmetries.
As an application, we show in Sec.~\ref{sec:QDprep} that inputting the symmetric product state into $\KW^G$ prepares the ground state of $\mathcal D(G)$, proving that the state can be prepared with $l_G$ rounds of measurements\footnote{Though it is worth noting that to specifically prepare the ground state of $\mathcal D(G)$ (rather than implementing the map $\KW^G$), it is possible for the number of shots to be lowered, as we have demonstrated for the nil-2 case where the derived length is two but the state can be prepared in a single shot.}.

\subsection{Review of KW map for $G$}\label{sec:KWG}

The Kramers-Wannier map for $G$ in $d$-spatial dimensions ($\KW^G$) is defined as a non-local transformation that maps between states with $G$ $0$-form symmetry to states in a $G$ gauge theory with $\Rep(G)$ $(d-1)$-form symmetry \cite{Yoshida17, Lootens21Duality}. Note that unlike Abelian groups, there is no natural generalization of $\KWd$, since there is no isomorphism between $G$ and $\Rep(G)$.

We begin with an arbitrary directed graph with $G$ degrees of freedom $\CC[G]$ on both vertices and edges with basis vectors given by group elements $\ket{g}_v$ and $\ket{g}_e$. At the level of states, the KW map acts by mapping
vertex degrees of freedom to edge degrees of freedom by heuristically mapping
spin variables on vertices to domain-wall variables on edges. The output of an edge $g_e$ is determined by the two vertices at the boundary where the initial and final vertices are denoted $i_e$ and $f_e$ respectively. That is $\KW^G: \mathbb C[G]^{\otimes N_v} \rightarrow \mathbb C[G]^{\otimes N_e}$ which acts as 
\begin{align}
\bigotimes_v \ket{g_v}_v \xmapsto[]{\KW^G} \bigotimes_e \ket{\bar g_{i_e} g_{f_e}}_e  
\label{equ:KWonstatesG}
\end{align}
where $\bar g$ denotes the inverse of $g$. Based on this definition, we can work out how operators map . First, we define the generalization of Pauli $X$ and $Z$ to finite groups\cite{Brell2015,Albert21}. The Pauli $X$ is generalized to left and right multiplication for each group element $g$
\begin{align}
L^g \ket{h} &= \ket{gh}, & R^g\ket{h} = \ket{h\bar g}.
\end{align}
Let $\bs \rho^\mu$ be the matrix representation for each irrep $\mu$ of $G$. 
The generalization of Pauli $Z$ is no longer a single site operator, but now a matrix product operator with bond dimension $d^\mu$, the dimension of the irrep $\mu$. Namely,
\begin{align}
   \bs Z^\mu_{ij} \ket{g} = \bs  \rho^\mu(g)_{ij} \ket{g}
\end{align}
where $i,j = 1,\ldots, d^\mu$ denote the bond indices of $\bs Z$. Here notation-wise, we also use a bold font to denote that the bond indices have not been contracted, and physical operators must be defined by fully contracting such indices. 

The defining property of the KW map is that it is not unitary. More specifically, it has non-trivial right and left kernels, which we denote as the ``symmetry" and ``dual symmetry". respectively. From the map \eqref{equ:KWonstatesG} it is apparent that performing left multiplication on all vertices will leave the output state invariant, and any product of $\bs Z$ around a closed loop is unchanged by an arbitrary input state. In equations,
\begin{align}
    \KW^G \times \left(\prod_v L^g_v\right)  = \KW^G\label{eq:G0form}\\
   \Tr \left[\prod_{e \in l}  \bs  Z^{\mu^{O_e}}_e\right] \times \KW^G = d^\mu\KW^G 
    \label{eq:repGd-1form}
\end{align}
The former is precisely  the $G$ 0-form symmetry, while the latter is given for an arbitrary closed loop $l$ ($\partial l=0$) and $O_e$ denotes the orientation of $e$ with respect to the loop, which conjugates a given representation if the orientation of $e$ goes against that of the loop. If $l$ is contractible, then it defines a local constraint, which can be thought of as a gauge constraint, while for non contractible loops, this defines a $\Rep(G)$ $(d-1)$-form symmetry. For non-Abelian $G$, this symmetry operator is not unitary and not onsite if $d^\mu > 1$, and is therefore a generalized notion of symmetry called a non-invertible or categorical symmetry (see \cite{Mcgreevy22,Cordova22Snowmass} for a review and further references). In this case, the $\Rep(G)$ symmetry has an intuitive interpretation as string operators whose end points are gauge charges for the quantum double of $G$ \cite{Bultinck17}. The fusion rules for such symmetries correspond precisely to the fusion rules for irreps of $G$.

Similar to the cluster state for Abelian groups, it is possible to represent $\KW^G$ as a tensor-network operator
\begin{align}
    \KW^G_{EV} = \bra{+}^G_V U_{EV}^G \ket{1}^G_E
\end{align}
where $U_{EV}^G $ is a unitary that generalizes the cluster state entangler to $G$ degrees of freedom\cite{Brell2015} (see Appendix~\ref{app:Gclustestate} for further details). However, we run into a problem when we try to implement the projection via measurement. In general, a complete set of measurement outcomes on each vertex is given by all irreps of the group $G$. Thus, suppose that the measurement outcome transforms under the operators $\prod_v L^g_v$ as some irrep $\mu$ on vertex $v$ and $\mu'$ on vertex $v'$, then the desired state where the measurement outcome is $\ket{+}$ can be recovered if we first act with the operator $ \bs Z^{\bar \mu}_v \bs Z^{ \mu}_{v'}$. Pushing this through the KW gives the output string operator we must apply to pair-annihilate the irreps and fix the state, which is $\prod_{e \in l} \bs Z^{\bar \mu}_e$
where $l$ is a path whose endpoints are $v$ and $v'$ (see Eq.~\eqref{eq:ZZmapG} for a derivation). If $d^\mu=1$, then the excitation is an Abelian anyon and the string operator factors to each edge and so it can be implemented in a single layer. However, if $d^\mu >1$, then the operator remains a matrix product operator that cannot be implemented in finite depth. 

Consequently, we now address how to overcome this problem for solvable groups by sequentially apply a KW map that gauges a sequence of Abelian subgroups in $G$, as given in Eq.~\eqref{eq:sequentialgauging}.

\subsection{KW for gauging a normal subgroup $N$ of $G$}\label{sec:KWNinG}

We proceed to define $\KW^{N\triangleleft G}$, and show it satisfies the two-step gauging property Eq.~\eqref{eq:2step}. Let us remark that even though $N$ can be assumed to be Abelian for the purposes of this paper, the formulas we present hold for an arbitrary (not necessarily Abelian) normal subgroup of $G$.

\subsubsection{Definition of $\KW^{N \triangleleft G}$}

First, let us define the maps that gives the $N$ and $Q$ components of an element in $G$:
\begin{align}
    \pi:& \ G \rightarrow Q , &\pi(g) &= q\\
    t:&  \ G\rightarrow N, & t(g)&=n
\end{align}
Namely, $g$ can be represented as $g=(n,q) =(t(g),\pi(g))$. Note that while $\pi$ is a group homomorphism, $t$ is not.

We define the $\KW^{N\triangleleft G}$ map by projecting to the $Q$ spins on the vertices via $\pi$ and to the $N$ domain walls via $t$. That is, $\KW^{N\triangleleft G}: \mathbb C[G]^{\otimes N_v} \rightarrow \mathbb C[N]^{\otimes N_e} \otimes \mathbb C[Q]^{\otimes N_v} $ given by

\begin{align}
\bigotimes_v \ket{g_v}_v \xmapsto[]{\KW^{N\triangleleft G}} \bigotimes_e  \ket{t(\bar g_{i_e} g_{f_e})}^N_e  \bigotimes_v \ket{\pi(g_v)}^Q_v 
\label{equ:KWonstatesNG}
\end{align}
Let us show that using this definition, the two-step gauging property~\eqref{eq:2step} holds. Further applying $\KW^Q$,  we find
\begin{align}
& \bigotimes_e  \ket{t(\bar g_{i_e} g_{f_e})}^N_e  \bigotimes_v \ket{\pi(g_v)}^Q_v \\
&\xmapsto[]{\KW^{Q}} 
\bigotimes_e  \left[\ket{t(\bar g_{i_e} g_{f_e})}^N_e \ket{\pi(\bar g_{i_e}g_{f_e})}^Q_e \right] =\bigotimes_e \ket{\bar g_1g_2}^G_e \nonumber
\end{align}
where we have used the fact that $\pi(\bar g_1) \pi(g_2) = \pi(\bar g_1g_2)$, and combined the $N$ and $Q$ degrees of freedom to a single $G$ degree of freedom $\bar g_1g_2$. Comparing to Eq.~\eqref{equ:KWonstatesG}, the two maps combined is exactly $\KW^G$.

\subsubsection{Physical derivation of $\KW^{N\triangleleft G}$}\label{sec:Uphysical}

Given the above formula, a direct computation that we perform in Appendix~\ref{app:KWN} shows that Eq.~\eqref{eq:KWNGdef} holds. Namely $\KW^{N\triangleleft G}$ can be implemented by $\KW^N$ followed by an extra FDLU $\mathcal U^{N\triangleleft G}$.

Here, we opt to motivate physically why this extra unitary is needed, and why it takes the above form. First, let us see what goes wrong in the absence of this unitary. Consider implementing only the map $\KW^N$
\begin{align}
\bigotimes_v \ket{g_v}_v \xmapsto[]{\KW^{N}} \bigotimes_e  \ket{\bar n_{i_e}n_{f_e}}^N_e  \bigotimes_v \ket{\pi(g_v)}^Q_v 
\label{equ:KWonstatesN}
\end{align}
Using this map, there will be a dual $\Rep(N)$ $(d-1)$-form symmetry defined as
\begin{align}
    \Tr \left[\prod_{e \in l}   {\bs Z}^{\nu^{O_e}}_e \right] \label{eq:KWNRepNsymm}
\end{align}
for every $\nu \in \Rep(N)$. However, the remaining $Q$ 0-form symmetry will not take the form $\prod_v L^q_v$ for a non-trivial extension. This can be confirmed explicitly by noting that performing left multiplication by $g$ on all the vertices does not leave the output state invariant because of the modified group multiplication rule \eqref{eq:multiplicationlaw}. Indeed, one finds
\begin{align}
\bigotimes_v \ket{gg_v}_v \xmapsto[]{\KW^{N}} \bigotimes_e  \ket{n'}^N_e  \bigotimes_v \ket{\pi(g_v)}^Q_v 
\label{equ:KWonstatesNGgaction}
\end{align}
where $n' =\overline{n\sigma^q[n_1]\omega(q,q_1)}n\sigma^q[n_2]\omega(q,q_2) \ne \bar n_1n_2$.

In fact, the physical reason why this must be the case when the group extension is non-trivial is because end points of $\Rep(N)$ (which can be thought of as anyons formed by the end points of the symmetry lines) must transform non-trivially under $Q$ \cite{BarkeshliBondersonChengWang2019,Hermele2014,TarantinoLinderFidkowski2016}
\begin{enumerate}
    \item When $\sigma$ is non-trivial, $Q$ must permute the $\Rep(N)$ $(d-1)$-form symmetry\footnote{If $\omega$ is trivial, this can be thought of as a split $d$-group}.
    \item When $\omega$ is non-trivial, $Q$ and $\Rep(N)$ have a mutual anomaly\cite{Tachikawa_2020}. (If $\sigma$ is trivial, this can be detected by symmetry fractionalization: the end points of $\Rep(N)$ will carry a projective representation under $Q$)
\end{enumerate}

For this reason, we need to further perform a basis transformation (in the Heisenberg picture) in order to turn $Q$ into an onsite symmetry $\prod_v L^q_v$ (at the cost of also modifying the form of $\Rep(N)$), so that $Q$ can be sequentially gauged in the next step.

To gain further insight into the required basis transformation, we now turn to investigate the  kernels of the map $\KW^{N \triangleleft G}$. First, we note that left multiplication by $N$ on all vertices leaves invariant $\pi(g_v)$ on all vertices and $t(\bar g_{i_e} g_{f_e})$ on all edges. It is therefore a right kernel of $\KW^{N\triangleleft G}$. More generally, this means that the $G$ symmetry will be reduced to a $Q$ symmetry under the $\KW^{N \triangleleft G}$ map
\begin{align}
    \KW^{N\triangleleft G} \times \prod_v L^g_v  = \prod_v L^q_v \times \KW^{N\triangleleft G} \label{eq:Q0form}
\end{align}
which is imperative for sequential gauging. On the other hand, the dual symmetry is not the usual $\Rep(N)$ symmetry defined in Eq.~\eqref{eq:KWNRepNsymm}. Instead, it is obtained by taking a product of $ \tilde  {\bs Z}^\nu_e$ operators around closed loops where
\begin{align}
    \tilde  {\bs Z}^\nu_e \ket{q_{i_e},n_e,q_{f_e}} &= \bs \rho^\nu(\tilde n_e) \ket{q_{i_e},n_e,q_{f_e}}, \label{eq:tildeZ}\\
      \tilde n_e &=  \aut{q_{i_e}}{n_e}\omega( q_{i_e}, \bar q_{i_e}q_{ f_e})
      \label{eq:ntilde}
\end{align}
for irreps $\nu \in \Rep(N)$. Namely, we have a $(d-1)$-form $\Rep(N)$ symmetry 
\begin{align}
      \Tr \left[\prod_{e \in l}   \tilde {\bs Z}^{\nu^{O_e}}_e \right] \times \KW^{N\triangleleft G} = \KW^{N\triangleleft G}
    \label{eq:repNd-1form} 
\end{align}

To verify this, we note that the group element on each edge can be expressed using the factor system as
\begin{align}
  n_e =   t(\bar g_{i_e}g_{f_e}) = \bar \omega(\bar q_{i_e},q_{i_e}) \aut{\bar q_{i_e}}{\bar n_{i_e} n_{f_e}}  \omega(\bar q_{i_e},q_{f_e})
    \label{eq:tgigfbar}
\end{align}
Inserting this into Eq.~\eqref{eq:ntilde}, and simplifying using the cocycle condition Eq.~\eqref{equ:cocyclecondition} gives
\begin{align}
  \tilde n_e = \bar n_{i_e}n_{f_e}
\end{align}
 Thus, taking a product around a closed loop, the contributions from each vertex pairwise cancels.

The difference between $\tilde{\bs Z}^\nu$ and $\bs Z^\nu$ allows us to back out the required basis transformation. Namely, $\mathcal U^{N \triangleleft G}$ must be the unitary such that
\begin{align}
   (\mathcal U^{N \triangleleft G})^\dagger \ket{q_{i_e}} \ket{n_e} \ket{q_{f_e}} = \ket{q_{i_e}} \ket{\tilde n_e} \ket{q_{f_e}}
\end{align}
so that $\tilde{\bs Z}^\nu = \mathcal U^{N \triangleleft G} \bs Z^\nu(\mathcal U^{N \triangleleft G})^\dagger$. From the definition of $\tilde n$ in Eq.~\eqref{eq:ntilde}, we see that Eq.~\eqref{eq:UasFDLU} is exactly the FDLU that does the job.

Let us now offer a complimentary viewpoint of $\mathcal U^{N \triangleleft G}$ in terms of the Schr\"odinger picture. That is, as a basis transformation on states instead of the symmetry operators. Suppose we input a symmetric product state of $G$, and applied $\KW^N$, but instead of declaring the $Q$ symmetry to be the result of pushing $\prod_v L^g_v$ through the KW map, we \textit{insisted} that the $Q$ 0-form symmetry is given by $\prod_v L^q_v$. Then at this point, the state we have prepared is the $N$ quantum double which is enriched \textit{trivially} by the symmetry $Q$. That is, if we now further gauge this onsite $Q$, the resulting state would be $\mathcal D(N\times Q)$, which would be the same result had the group extension been trivial\footnote{this corresponds to $\KW^Q\KW^N = \KW^{N \times Q}$}. Thus, to fix this we need to entangle the state into a \textit{non-trivial} SET. This entangler is precisely $\mathcal U^{N \triangleleft G}$! The two layers of $\mathcal U^{N \triangleleft G}$ depends on the two data specifying the group extension, and serves the following roles that enrich the state:
\begin{enumerate}
\item $\Sigma_{EV}$ has an action from the vertices to the edges, which makes the gauge charges (end points of $\Rep(N)$) permute correctly under the action of the $Q$ 0-form symmetry.
\item $\Omega_{VEV}$ can be viewed as the entangler that decorates the strings of the gauge charges with a 1+1D ``SPT state"\footnote{``SPT" appears in scare quotes here because the cocycles are valued in $N$ rather than $U(1)$. A fractionalization class does not always correspond to a 1+1D SPT phase.} given by the cocycle $\omega$ which gives the end points the correct symmetry fractionalization by $Q$.
\end{enumerate}

In Appendix \ref{app:KWN}, we prove Eq.~\eqref{eq:KWNGdef} explicitly by showing that the right hand side has the correct kernels as shown in Eqs. \eqref{eq:Q0form} and \eqref{eq:repNd-1form}.

\subsection{Preparation of $\mathcal D(G)$}\label{sec:QDprep}

As an application, we can use $\KW^G$ to prepare the ground state of $\mathcal D(G)$ by applying it on a symmetric product state~\cite{Brell2015,Yoshida17}

\begin{lemma}
\begin{shaded}
  For any finite group $G$, $\ket{ \mathcal D(G)}_E = \KW^G_{VE} \ket{+}^G_V$. 
  \end{shaded}
\begin{proof}
   Recall that $\ket{+}_V$ is the $+1$ eigenstate of the projectors
\begin{align}
    R_v &= \ket{+}^G\bra{+}^G= \frac{1}{|G|}\sum_g R^g_v
\end{align}
For each vertex. In Appendix \ref{app:Gclustestate}, we show that
\begin{align}
   R^g_v \xrightarrow[]{KW^G}  A_v^g =\prod_{e\rightarrow v} R^g_e\prod_{e\leftarrow v} L^g_e.
   \label{eq:LgtoAg}
\end{align}
Therefore, the output state must satisfy
\begin{align}
    A_v = \frac{1}{|G|} \sum_g A_v^g =\prod_{e\rightarrow v} R^g_e\prod_{e\leftarrow v} L^g_v =1,
\end{align}
which is the vertex term of the quantum double model. In addition, for each closed loop around a plaquette $p$, the left kernel of $\KW^G$ tells us that the state also satisfies
\begin{align}
     B_p^\mu =\Tr\left [\prod_{e\in \partial p}  \bs Z_e^{\mu^{O_e}}\right]  =d^\mu
\end{align}
for each irrep $\mu$. Summing over all irreps weighted by their dimensions and using $\frac{1}{|G|}\sum_\mu d^\mu \chi^\mu(g) =\delta_{1,g}$ and $\sum_\mu (d^\mu)^2 = |G|$ we have
\begin{align}
     B_p =\frac{1}{|G|} \sum_\mu d^\mu   B_p^\mu = \sum_{\{g_e\}}\delta_{1,  \prod_{e \subset p}  g_e^{O_e}} \ket{\{g_e\}}\bra{\{g_e\}} =1,
\end{align}
which is precisely the plaquette term of the quantum double model. This shows that the state $\KW^G_{VE} \ket{+}_V$ has eigenvalues $+1$ under both $A_v$ and $B_p$.
\end{proof}
\end{lemma}

Combining this with Theorem~\ref{thm:Gsolvable}, which gives an explicit protocol to implement $G$ with measurement for solvable groups, we have the following result:

\begin{cor}
\begin{shaded}
    For a solvable group $G$ with derived length $l_G$, the ground state of $\mathcal D(G)$ can be prepared with FDLU, $l_G$ rounds of measurement and feedforward
    \end{shaded}
\end{cor}

Similarly, one can input a $G$-SPT instead of a symmetry product state. The resulting state will be a twisted quantum double of $G$ \cite{DijkgraafPasquierRoche91TQD,HuWanWu13TQD}. Again, we remind the reader that to prepare $\mathcal D(G)$ the measurement rounds can be lowered, as we have shown for nil-$2$ groups.

\section{No-go for preparation by a  finite number of shots: Fibonacci anyons and non-solvable quantum doubles}\label{sec:nonsolvable}

So far, we have demonstrated that there is an interesting hierarchy of states and maps depending on the number of shots required to prepare or implement it, which is summarized in Table~\ref{table:hierarchy}. It is equally valuable to know negative results, similar to how it is interesting to note that, say, the toric code \emph{cannot} be obtained from the product state by an FDLU \cite{BravyiHastingsVerstraete06}. Note that acting with an FDLU followed by measurement on a state naturally gives rise to a projected entangled pair state (PEPS) representation of the wavefunction. Thus, this automatically excludes volume law states. It is also widely believed that chiral states do not admit a PEPS representation with finite bond dimension\cite{kitaev_anyons_2005,DubailRead15,FK_CommProj2019,SK_CommProj2020} (however this restriction can sometimes be overcome by using gates with exponential tails; see Sec.~\ref{sec:qFDLU}). On the other hand, there are a wide range of states in two dimensions that admit a PEPS representation. In fact, it was recently shown that even certain critical states admit such a representation, and particular ensembles of them can be efficiently prepared in this manner \cite{Nishimoricat22,Lee22}.

In this Section, we nevertheless argue that there are certain phases---in fact, even fixed-point states which admit relatively simple PEPS representations---that \emph{cannot be prepared in finite time} using FDLU, measurement and feedforward. Namely, we argue that a finite number of shots cannot prepare non-solvable quantum doubles and the Fibonacci topological order (Fib)\footnote{To make the latter claim a non-trivial statement, one has to allow quasi-FDLU, since Fib is chiral (see Sec.~ \ref{sec:qFDLU}); indeed, our arguments direclty extend to the case where the unitary gates have exponentially small tails. Nevertheless, a similar argument holds analogously for the double Fibonacci phase, which admits a PEPS representation.}.

\subsection{The necessity of creating nonlocal defects for measurement-prepared topological order}

Let us first recall one of the simplest cases of preparing topological order via measurement: one obtains the toric code upon measuring its stabilizers. The randomness of measurement gives us a speckle of `anyon defects'. While these can be paired up with only a single layer of feedforward gates (to deterministically prepare the clean toric code), this is a conditional gate that depends \emph{nonlocally} on the measurement outcome. Here we formalize the intuition that this is unavoidable by proving the following theorem regarding states that can be prepared with measurement and local corrections (i.e., where one does not require nonlocal classical communication upon applying feedforward):

\begin{theorem}
\label{thm:thm2}
\begin{shaded}
If a state $\ket{\psi_\textrm{out}}$ is deterministically obtained from an input state $\ket{\psi_\textrm{in}}$ by single-site measurements followed by local corrections implemented by an FDLU, then there exists a state $\ket{\varphi}$ such that $\ket{\psi_\textrm{in}}$ is related by an FDLU to $\ket{\psi_\textrm{out}} \otimes \ket{\varphi}$.
\end{shaded}

\begin{proof}
Denote the total Hilbert space $\mathcal H$, the subspace on which the measurements are performed $\mathcal H_M$, and the remaining Hilbert space $\mathcal H_{\overline{M}}$. For each measurement outcome $\ket{s_i}$ where $i\in M$, we follow up by a local unitary $U_i(s_i)$ that depends on $s_i$. All $U_i(s_i)$ must commute regardless of the measurement outcome on each site, since the corrections cannot depend on the order in which they are applied. For a given measurement outcome, we can therefore write the resulting state as
\begin{align}
  \prod_{i \in M}  U_i(s_i)  \ket{s_i}\bra{s_i} \times \ket{\psi_\textrm{in}}.
\end{align}
Next, define the controlled operator
\begin{align}
    CU_i = \sum_{\{s_i'\}} \ket{s_i'}\bra{s_i'}  U_i(s_i')
\end{align}
which applies $ U_i(s_i')$ depending on an orthonormal basis of measurement outcomes $\{s_i'\}$. Orthonormality guarantees that $CU_i$ is unitary. Moreover, since all $U_i(s_i)$ commute, all the control gates must also commute. We next note that
\begin{align}
    U_i(s_i)\ket{s_i}\bra{s_i} = \ket{s_i}\bra{s_i} CU_i 
\end{align}
Note that here $U$ does not act on the ancillas (indeed, the feedforward only needs to correct on $\mathcal H_{\overline{M}}$ to deterministically prepare $\ket{\psi_\textrm{out}}$).
Therefore, the resulting state may also be expressed as
\begin{align}
  \prod_{i \in M} \ket{s_i}\bra{s_i} CU_i  \times \ket{\psi_\textrm{in}}.
\end{align}
That is, this is a deterministic unitary followed by single site measurements. Now, if this always gives the state $\ket{\psi_\textrm{out}}$ on $\mathcal H_{\overline{M}}$, this means that the measurement at the end must not affect the output state. Therefore, before the measurement, we must have
\begin{align}
 \ket{\psi_\textrm{out}}_{\mathcal H_{\overline{M}}} \otimes \ket{\varphi}_{\mathcal H_M}  = \prod_{i \in M} CU_i  \times \ket{\psi_\textrm{in}}_{\mathcal H},
\end{align}
for some state $\ket{\varphi}$. Since $\prod_{i \in M} CU_i$ is an FDLU, this completes the proof.
\end{proof}
\end{theorem}

\begin{cor}
\begin{shaded}
Starting with a product state, a single-shot protocol with locally correctable outcomes can only prepare invertible states. 
\label{thm:cor3}
\end{shaded}
\begin{proof}
    Recall that invertible states $\ket{\psi_\text{inv}}$ are states such that there exists an ``inverse state" $\ket{\psi_\text{inv}^{-1}}$ such that $U \ket{\psi_\text{inv}} \otimes \ket{\psi_\text{inv}^{-1}} $ is a product state for some FDLU $U$. By definition, if $\ket{\psi_\text{in}}$ is a product state, then $\ket{\psi_\text{out}}$ is an invertible state.
\end{proof}
\end{cor}
Likewise, for a finite-shot protocol, one can still prepare only invertible states from the product state if at each level the measurement results are locally correctable. Moreover, the same reasoning suggests that in the space of maps, locally-correctable protocols can implement only QCAs.

\subsection{No-go argument for the Fibonacci topological order: `Fib' is a non-trivial measurement-equivalent phase}

Let us use the above result to present a plausibility argument that Fib cannot be prepared in a finite number of shots. We will assume there is a sequence of FDLU, measurements and feedforward corrections that prepares Fib, and argue a contradiction. Let us focus on the final shot of a potential multi-shot protocol.

First, let us assume that in the final shot there are measurements that are non-locally correctable. Then, the obtained wavefunction can be realized as a ground state of a alternate Hamiltonian that differs from the true Hamiltonian by terms localized near such measurement outcomes. Such a wavefunction can be said to contain ``defects". Such defects can be classified by non-trivial superselection sectors, since by definition, they cannot be removed by local corrections. Such superselection sectors are referred to as \textit{anyons} if they are zero-dimensional, and \textit{line defects} if they are one-dimensional in space. However, Fib has only one type of non-Abelian anyon $a$ with fusion rules $a\times a = 1+ a$. It also does not have any line defects, since the automorphism group is trivial, and it does not permit a non-trivial symmetry fractionalization class. Therefore, if we wish to deterministically prepare the state, we are not allowed to postselect on obtaining the scenario where all measurement outcomes are locally correctable (corresponding to the trivial anyon $1$). Hence, it is possible that some measurement outcome results in a non-locally correctable error, corresponding to the anyon $a$. Since, $a$ is non-Abelian, it cannot be paired up in finite depth \cite{Shi19,Liu21,Bravyi22}. Thus, feedforward corrections will fail to prepare Fib.

The only remaining possible scenario is then that all measurements are locally correctable. By Theorem \ref{thm:thm2}, this is only possible if the state before the measurement is equivalent (up to an FDLU that can potentially depend on the measurement outcome of the previous round) to $\ket{\textrm{Fib}} \otimes \ket{\mathcal C}$ for some state $\ket{\mathcal C}$.
Let us now ask whether $\ket{\textrm{Fib}} \otimes \ket{\mathcal C}$ can itself be prepared by a finite number of shots. Again, measurement outcomes can correspond to anyons or line defects in the phase $\textrm{Fib} \boxtimes \mathcal C$. The only possible Abelian anyons are those entirely in $\mathcal C$, and since the states are in a tensor product, measuring these Abelian anyons cannot help prepare $\ket{\textrm{Fib}}$. Thus, we are left with the possibility that the measurements result in line defects in $\ket{\textrm{Fib}} \otimes \ket{\mathcal C}$. By a similar argument, if the line defects are non-invertible, they cannot be shrunken away with an FDLU. Therefore, we only consider the case that the measurement outcomes correspond to invertible line defects.

We now give a physical argument that measuring invertible line defects does not help prepare topologically ordered states. Since these defects correspond to charges of 1-form symmetries \cite{generalizedglobalsymmetries}, if they occur as measurement outcomes, they would come from a KW map that gauges a 1-form symmetry, which physically corresponds to anyon condensation. Thus, measuring such line defects only serves to reduce the topological order, rather than creating a more complicated one. Thus the parent state itself should be even harder to prepare. Indeed, from the category theory point of view, it is known that the Fibonacci topological order cannot be trivialized by a finite number of Abelian gauging procedures \cite{Etingof2011weakly}.

To give an explicit example, consider $\mathcal C$ to be another copy of Fib. Then there is an Abelian line defect given by a $\ZZ_2$ symmetry that swaps the two copies. However, the parent phase on which measuring would give such a line defect corresponds to the phase obtained by gauging the SWAP symmetry. Thus, this does not help us prepare Fib.

To conclude, we conjecture the following:
\begin{conjecture}
\label{thm:fib}
\begin{shaded}
The Fibonacci topological order cannot be prepared deterministically from the product state by a finite number of shots.
\end{shaded}
\end{conjecture}

Assuming that the above conjecture is true, this also implies that any number of copies of Fibonacci also cannot be prepared. If we were able to prepare $N$ copies, we can simply ``measure away" $N-1$ copies and we would be left with a single copy of Fib. Similarly, the double Fibonacci topological order, which admits a PEPS wavefunction\cite{levin_string-net_2005,GuLevinSwingleWen09,Buerschaper09}, cannot be prepared in any number of shots.

This result motivates us to define the notion of a \textit{measurement-equivalent phase}.
\begin{definition}
\begin{shaded}
Two states $\ket{\psi}$ and $\ket{\psi'}$ are in the same \emph{measurement-equivalent phase} if one can deterministically prepare both $\ket{\psi}$ from $\ket{\psi'}$ and $\ket{\psi'}$ from $\ket{\psi}$ using a finite number of rounds of FDLU, measurements and (finite-depth) feedforward.
\end{shaded}
\end{definition}

Our claim implies that Fib realizes a non-trivial measurement-equivalent phase, 
distinct from twisted quantum doubles for $D(G)$ for solvable groups, which lie in the trivial measurement-equivalent phase.

\subsection{Conjecture for non-solvable quantum doubles}

Similarly, we can construct a similar argument for non-solvable quantum doubles. In this case, it first helps to show that certain groups are in the same measurement-equivalent phase.

\begin{lemma}
\begin{shaded}
The quantum double $\mathcal D(G)$ for any finite group $G$ is in the same measurement-equivalent phase as $\mathcal D(G_{pc})$, where $G_{pc}$ is a perfect centerless group corresponding to the central quotient of the perfect core of $G$.
\end{shaded}
\end{lemma}

Here, we recall a few definitions. A perfect group is defined as a group that is equal to its own commutator subgroup $(G=[G,G])$, and the perfect core $G_p$ of a group $G$ is the largest perfect subgroup of $G$. A group $G$ is centerless if its center $Z(G)$ is trivial, and the central quotient of a group $G$ is defined as $G/Z(G)$. An example of a perfect centerless group is $A_n$, the alternating group on the set of $n$ elements for $n\ge 5$.

The above lemma tells us that we can reduce the problem to showing that $\mathcal D(G_{pc})$ realizes a non-trivial measurement-equivalent phase for each non-trivial $G_{pc}$. We note in particular that for solvable groups, their quantum doubles are in the same measurement-equivalent phase as that of the trivial group $G_{pc} = \{ 1 \} $. 

The following proof can be interpreted as the fact that one can sequentially condense Abelian anyons starting from $\mathcal D(G)$ to reach $\mathcal D(G_{pc})$. Notably, $\mathcal D(G_{pc})$ can be thought as a ``fixed point" in the measurement-equivalent phase because it only contains non-Abelian anyons.

\begin{proof}
To show our claim, we turn to the derived series \eqref{eq:derivedseries} of an arbitrary finite group $G$. The derived series will always stabilize to the perfect core $G_p$ of $G$. Therefore, given such a group, one can always start from $\mathcal D(G_p)$ and apply the KW map to sequentially gauge appropriate groups according to the derived series in the same spirit as Eq. \eqref{eq:sequentialgauging} in order to arrive at $\mathcal D(G)$.

Specifically, we have
\begin{align}
\ket{\mathcal D(G)}_E  &= \KW^{G}_{EV} \ket{+}^{G}_V \\
&=  \KW^{G/G_p}_{EV} \mathcal U^{G_p \triangleleft G}_{EV} \KW^{G_p}_{EV}\ket{+}^{G_p}_V \ket{+}^{G/G_p}_V  \\
&= \KW^{G/G_p}_{EV} \mathcal U^{G_p \triangleleft G}_{EV} \ket{\mathcal D(G_p)}_E\ket{+}^{G/G_p}_V
\end{align}
since $G/G_p$ is solvable, we can prepare $\KW^{G/G_p}_{EV}$ using measurements and feedforward.

Next, we further argue that $\mathcal D(G_p)$ for any perfect group can in turn be prepared from $\mathcal D(G_{pc})$. If $G_p$ is already centerless then we are done. Otherwise, Gr\"un's lemma states that the quotient group $G_p/Z(G_p)$ is centerless, and is therefore a perfect centerless group. Thus, starting from $\mathcal D(G_{pc})$ we can turn it into an SET state where the fluxes are fractionalized by the symmetry $Z(G_p)$ according to the 2-cocycle that determines the central extension
\begin{align}
 1 \rightarrow Z(G_p)  \rightarrow G_p \rightarrow G_{pc} \rightarrow 1
\end{align}
Then, gauging $Z(G_p)$ using the KW map prepares $\mathcal D(G_{p})$ as desired.

Reciprocally, to prepare $\mathcal D(G_{pc})$ from $\mathcal D(G)$ it suffices to perform measurements of the hopping operators to condense the anyons that resulted from the gauging in the reverse order.
\end{proof}

Similarly, to the Fibonacci case, let us consider the defects in $\mathcal D(G_{pc})$. The anyons in $\mathcal D(G_{pc})$ are all non-Abelian: since the group is perfect, it does not have (non-trivial) one dimensional irreps (corresponding to gauge charges), and since the group is centerless, all conjugacy classes except the trivial one have more than one group element (corresponding to non-Abelian fluxes and dyons). Therefore, the only way to prepare this phase comes from measuring Abelian line defects. Again, as we have argued, this does not help since the parent state must be a larger topological order. We thus conjecture:

\begin{conjecture}
\begin{shaded}
The quantum doubles $\mathcal D(G_{pc})$ and $\mathcal D(G_{pc}')$ for distinct perfect centerless groups $G_{pc}$ and $G_{pc}'$ are in distinct measurement-equivalent phases.
\end{shaded}
\end{conjecture}

\section{Quasi-local unitaries and measurements}\label{sec:qFDLU}
So far, our discussion has focused on states that can be prepared using (strictly) local unitaries, by which we mean finite-depth circuits consisting of finite-range gates. In this final section, we discuss what we can obtain if we instead allow for quasi-local unitaries, i.e., unitary gates with exponentially small long-range tails.

First, using ancillas, it is now possible to prepare any invertible (possibly chiral) state since the ``doubled" phase (the phase along with its time-reversed partner) can always be prepared from the product state via quasi-LUs\cite{Barbarino20}. The partner can then be discarded leaving us with a single chiral state (see Appendix \ref{app:Chern} for an explicit example for the Chern insulator). More generally, on the space of maps rather than states, any quasi-local QCA can also be implemented using quasi-FDLU and ancillas. Indeed, we note that although the proof given in \cite{HaahFidkowskiHastings18} (that QCA $\otimes$ QCA$^{-1}$ is an FDLU) presumes the strictly local case, the proof carries over to the quasi-local case. We also note that in one dimension, quasi-local QCAs have the same classification as that of their local counterparts \cite{Ranard20}, while higher-dimensional classifications are unknown.

Adding measurements, one can then perform either KW or JW to gauge such invertible state in one shot \cite{measureSPT}. For example, the Ising topological order can be prepared with quasi-LU and one round of measurement by gauging a $p+ip$ superconductor. Similarly, all of Kitaev's 16-fold way can be similarly prepared. We note that Ref.~\cite{LuHsieh} observed that one can also obtain the 16-fold way states by locally measuring the parity operator, as a sort of measurement-induced Gutzwiller projection; this also qualifies as a single-shot protocol, although this viewpoint was not emphasized.

It is also worth noting that while Abelian anyon theories with gappable edges can be prepared with FDLU and single-site measurements (since they can be prepared by gauging an appropriate Abelian SPT phase \cite{KS11,Kaidi2021higher}), all cases with \emph{ungappable} boundaries can be prepared using quasi-FDLU in one shot. Given such an Abelian anyon theory $\mathcal M$,  its ``double" $\mathcal M \boxtimes \overline{\mathcal M}$ admits a gapped boundary, and therefore can be prepared in one shot. Then, $\mathcal M$ and $\overline{\mathcal M}$ can be separated to different Hilbert spaces using a quasi-local unitary so that we can now discard $\overline{\mathcal M}$. For example, the $\nu=1/3$ Laughlin state\cite{Laughlin83} can be prepared by first preparing a $\ZZ_3$ Toric Code, adding fermionic degrees of freedom, then performing a quasi-FDLU to the fixed point of the doubled Laughlin state.

Note the importance of quasi-FDLU in preparing states with ungappable edges. In the absence of exponentially small tails, our preparation scheme would have given a PEPS realization of such a state with finite bond dimension, which is believed to be impossible~\cite{kitaev_anyons_2005,DubailRead15,FK_CommProj2019,SK_CommProj2020}.

\section{Outlook}\label{sec:outlook}
In this work, we have introduced a hierarchy of long-range entangled states based on the number of shots required to prepare the state.  In particular, we provided an explicit protocol that shows that nil-2 quantum double states (despite being non-Abelian) are as simple to prepare as Abelian topological states when measurement is an additional resource. In general, we have presented a hierarchy of KW maps for solvable groups $G$ based on their derived length. Moreover, for groups with an infinite derived length (i.e., non-solvable groups) we conjecture that their quantum doubles are in distinct measurement-equivalent phases of matter, and similarly for the Fibonacci topological order.

It is interesting to make a comparison to the ability for such states to be universal for quantum computation. Indeed, it is known that non-nilpotent solvable quantum doubles can only realize Clifford gates by braiding, whereas Fibonacci anyons and non-solvable quantum doubles can realize non-Clifford gates by braiding alone. Nevertheless, for non-nilpotent quantum doubles,  additional ancillas and measurement can enable a universal gate set amenable for topological quantum computation\cite{OgburnPreskill99,Kitaev_2003,Mochon04}. It is thus worth exploring whether there is a deeper connection between the hierarchy of states from measurements and the computational power of the prepared state. Moreover, we have pointed out that that solvable but non-nilpotent quantum doubles require at least two rounds of measurement. It would be interesting to see whether if there is any further increase in computational power (such as a denser universal gate set) for such quantum doubles that require at least three or more rounds of measurement.

Moreover, it is interesting to note that while the present work argues that Fib (and in particular double Fib) cannot be prepared by a finite number of shots, a recent work has shown that string-net models \emph{can} be prepared using $\mathcal O( \ln L)$ layers, where $L$ is the system size \cite{Lu2022} (which can be compared to the known linear depth protocols involving unitary circuits \cite{Liu21}). It is thus tempting to think that perhaps using $\sim \ln L$ shots is optimal, although this is unproven.

Looking forward towards the preparation of more general topological phases of matter in 2+1D, we believe our results generalize to anyon theories described by modular tensor categories (MTCs). Namely, we conjecture that all nil-2 MTCs\footnote{Nilpotence of a MTC is well-defined if we do not require the subcategories appearing in the upper central sequence to be modular.} \cite{GelakiNikshych08} can be prepared in one shot using quasi-FDLUs. In particular, this includes anyon theories that do not have a Lagrangian subgroup, such as Ising anyons. Similarly, we conjecture that all solvable MTCs \cite{Etingof2011weakly} can be prepared using quasi-FDLUs and a finite number of shots. A step towards proving this conjecture, as well as the conjectures given in Sec.~\ref{sec:nonsolvable} would be to show rigorously that performing measurements that are non-locally correctable relates the initial and final topological orders by gauging an Abelian symmetry. A further interesting question is whether all representatives of measurement-equivalent phases are given by perfect MTCs (theories that only contain non-Abelian anyons).

Regarding the preparation of solvable MTCs, it would be worthwhile to obtain rigorous results about the minimal number of shots required to prepare a given state of matter. For example, this number for quantum doubles $\mathcal D(G)$ is upper bounded by the derived length $l_G$, but can be lowered as we have shown for nil-2 groups. How does one calculate this number for general solvable MTCs and does this minimal number coincide with any interesting mathematical quantity?

Relatedly, the explicit form for $(\KW^G)^\dagger$ shows that $\Rep(G)$ can be gauged for any finite group $G$ (which does not need to be solvable) using only a single round of measurement. This is because the measurement outcomes correspond to domain walls of $G$, which can be paired up with $L^g$. In particular, this implies that symmetry broken phases of $G$ can be prepared for any finite $G$ in one shot. Are there other interesting non-invertible symmetries that can be gauged efficiently using measurements and feedforward?

Although the construction in the present paper applies to a system without a boundary, we believe that it is straightforward to apply the construction to the case with a boundary. First of all, applying the KW map to a system with a smooth boundary produces a particular gapped boundary, namely the boundary where all gauge fluxes condense. For example, measuring a 2D cluster state with boundary prepares the toric code where all the $m$-anyons condense. In two spatial dimensions, it is known that gapped boundaries of the quantum double of $G$ are classified by a subgroup $K$ of $G$ and a 2-cocycle $H^2(K,U(1))$\cite{beigi2011quantum,luo2022gapped}. This can be physically interpreted as a particular 1D $G$-symmetric state before gauging. Namely, the boundary corresponds to a symmetry-breaking state where the subgroup $K$ is preserved, and the remaining symmetry can be put in to an SPT state. Since such symmetry breaking and SPT states can also be prepared using FDLU and measurements, this gives an explicit way to construct solvable quantum doubles with arbitrary gapped boundaries. We leave the explicit construction of twisted quantum doubles with arbitrary gapped boundaries, and gapped boundaries of topological orders in higher dimensions (itself an active line of research \cite{zhao2022string,luo2022gapped,Ji2022boundary}) to future work. The preparation of topological orders with condensation defects inserted\cite{Bombin10Defect,KitaevKong12,YouJianWen13,roumpedakis2022higher} would also be an interesting direction.

In higher dimensions, gauge theories of nil-2 groups naturally generalize to higher group gauge theories \cite{KapustinThorngren2017} where Abelian $0$-form symmetries are ``centrally" extended by higher-form symmetries. We give a protocol to prepare such a class of 2-group gauge theories in 3+1D in one shot in Appendix \ref{app:2group}. One can also prepare hybrid fracton phases\cite{TJV1,TJV2} where an Abelian 0-form symmetry is centrally extended by subsystem symmetries in one shot.

\begin{acknowledgments}

The authors thank Ryan Thorngren for collaboration on a related project \cite{measureSPT}, NT and RV thank Aaron Friedman, Andy Lucas, and Drew Potter for illuminating discussions. NT thanks Wenjie Ji for helpful discussions on $\Rep(G)$ symmetry and Julia Plavnik for helpful discussions on nilpotent and solvable categories. NT is supported by the Walter Burke Institute for Theoretical Physics at Caltech. RV is supported by the Harvard Quantum Initiative Postdoctoral Fellowship in Science and Engineering, and RV and AV by the Simons Collaboration on Ultra-Quantum Matter, which is a grant from the Simons Foundation (651440, AV).
\end{acknowledgments}

\vspace{5pt}

\emph{Note added.} As this manuscript was being prepared, we learnt of an upcoming work \cite{Li22}, which provides an alternate protocol to prepare twisted quantum doubles for non-Abelian groups. In addition, the authors of Ref.~\onlinecite{Bravyi22} have recently informed us of a refinement to their previous proof, to appear, that removes the implicit restriction to split extensions in Ref.~\onlinecite{Bravyi22}.

\vspace{5pt}

\bibliography{bib.bib}
\onecolumngrid
\newpage
\appendix

\newpage

\section{Invertible states from ancillas}

\subsection{Preparing the Kitaev chain from FDLU and ancillas}\label{app:Kitaev}

The Kitaev chain can be prepared by FDLU and ancillas by preparing two Kitaev chains and discarding the second copy. We start with two trivial fermionic chains with Majorana operators $\gamma_n, \gamma_n'$ and $\eta_n, \eta_n'$. The stabilizer of the trivial state (atomic insulator) with all sites unoccupied is
\begin{align}
    -i\gamma_n\gamma_n' & & -i\eta_n \eta_n'
\end{align}
Now we implement two layers of Majorana swap gates
\begin{align}
    U = \exp\left(\frac{\pi}{4}\sum_n\eta_{n+1} \gamma_n' \right) \exp\left(\frac{\pi}{4} \sum_n  \eta_n\gamma_n'\right) 
\end{align}
From this we see that
\begin{align}
U\gamma_n U^\dagger &= \gamma_n & U\gamma_n' U^\dagger &= -\gamma_{n-1}'\\
U\eta_n U^\dagger &= -\eta_{n+1} & U\eta_n' U^\dagger &= \eta_n'
\end{align}
That is, all $\gamma'$ are translated to the left while all $\eta$ are translated to the right (up to a minus sign). The resulting stabilizers are 
\begin{align}
    i\gamma_n\gamma_{n-1}' & & i\eta_{n+1} \eta_n'
\end{align}
which is exactly the stabilizer for two copies of the Kitaev chain.

\subsection{Preparing the Chern Insulator from quasi-FDLU and ancillas}\label{app:Chern}

\begin{figure}[h!]
    \centering
    \begin{tikzpicture}
    \node at (0,0) {\includegraphics[scale=0.75]{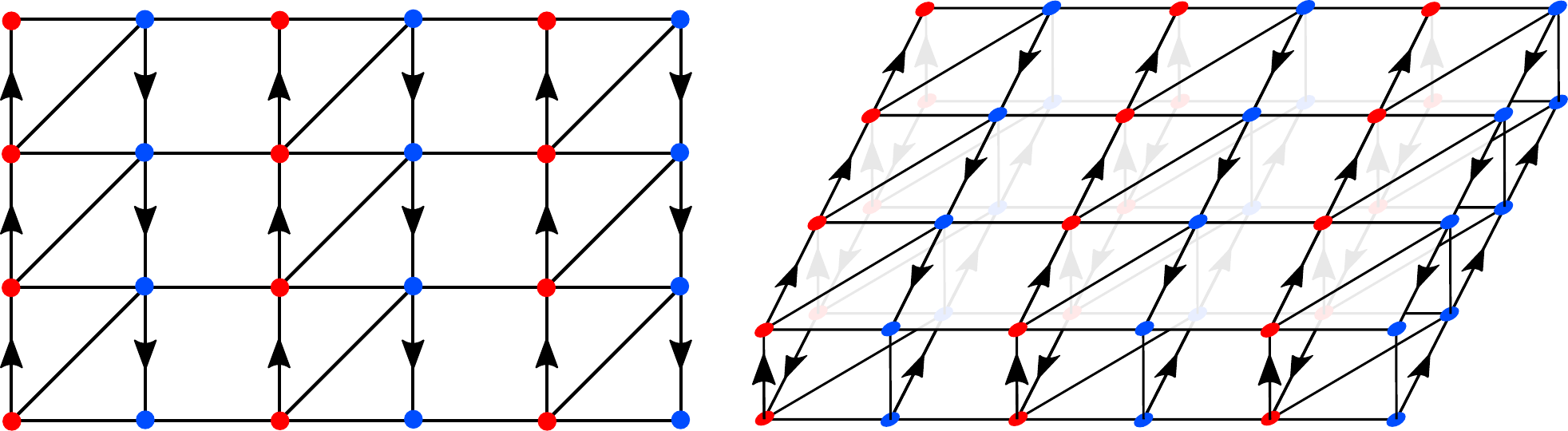}};
    \node at (-8,1.7) {(a)};
    \node at (0,1.7) {(b)};
    \end{tikzpicture}
    \caption{(a) 2D lattice model with nonzero Chern number (see Eq.~\eqref{eq:Chern}); bonds with arrows correspond to an imaginary hopping. (b) a stack of two Chern insulators with opposite Chern numbers; it can be adiabatically connected to a product state by tuning the rung couplings.}
    \label{fig:chernstack}
\end{figure}
Next, outline how to prepare the Chern Insulator with quasi-FDLU and ancillas. First, consider the following Hamiltonian
\begin{equation}
H = \sum_{x,y} \left( \frac{i}{2} a_{x,y+1}^\dagger a_{x,y}^{\vphantom \dagger} - \frac{i}{2} b_{x,y+1}^\dagger b_{x,y}^{\vphantom \dagger} + a_{x,y}^\dagger b_{x,y+1}^{\vphantom \dagger}+ a_{x+1,y}^\dagger b_{x,y}^{\vphantom \dagger} + a_{x,y}^\dagger b_{x,y}^{\vphantom \dagger}  \right) +h.c. \label{eq:Chern}
\end{equation}
where the couplings are depicted in Fig.~\ref{fig:chernstack}(a). The red (blue) dots are $A$ ($B$) sites, forming the two-site unit cell.
It can be straightforwardly checked that this has Chern number $\mathcal C=1$.  In momentum space\footnote{$a_{x,y} =\frac{1}{2\pi} \int \mathrm d k_x \mathrm d k_y e^{ik_x x + i k_y y} a_{k_x,k_y}$}:
\begin{equation}
H =  \iint \mathrm dk_x \mathrm dk_y \left( a_{k_x,k_y}^\dagger, b_{k_x,k_y}^\dagger \right)  \left( \begin{array}{cc}
\sin(k_y) & 1+ e^{-ik_x} + e^{ik_y} \\
1+ e^{ik_x} + e^{-ik_y} & -\sin(k_y)
\end{array} \right)
\left(\begin{array}{c}
a_{k_x,k_y} \\
b_{k_x,k_y}
\end{array}\right)
\end{equation}

The non-trivial Chern number implies that the ground state cannot be adiabatically connected to a product state. However, let us consider the double stack in Fig.~\ref{fig:chernstack}(b). where a second copy is introduced with inverted signs of the imaginary hopping, giving  $\mathcal C =-1$ to the second copy. As a whole, this non-chiral system can be adiabatically connected to a product state by simply introducing rung couplings; the momentum-space Hamiltonian is:
\begin{equation}
\mathcal H_{\boldsymbol k} = (1-\lambda) \left( \begin{array}{cc|cc}
\sin(k_y) & 1+e^{-ik_x} + e^{i k_y} & 0 & 0 \\
1+ e^{ik_x} + e^{-ik_y}. & -\sin(k_y) & 0  & 0 \\ \hline
0 & 0 & -\sin(k_y) &1+e^{-ik_x} + e^{i k_y}\\
0 & 0 &1+e^{ik_x} + e^{-i k_y} &  \sin(k_y) 
\end{array} \right)
+ \lambda\left( \begin{array}{cc|cc}
0 & 0 & i & 0 \\
0 & 0 & 0 & 1 \\ \hline
-i & 0 & 0 & 0\\
0 & 1 & 0 & 0
\end{array} \right)
\end{equation}
For $\lambda=0$ we have the two decoupled layers, but for $\lambda=0$ we add a real (imaginary) hopping between the red (blue) sites. The case $\lambda=1$ is a trivial product state along each rung. The ground state remains gapped throughout, as shown in Fig.~\ref{fig:chernstackgap}.

\begin{figure}[t!]
    \centering
    \includegraphics[scale=0.3]{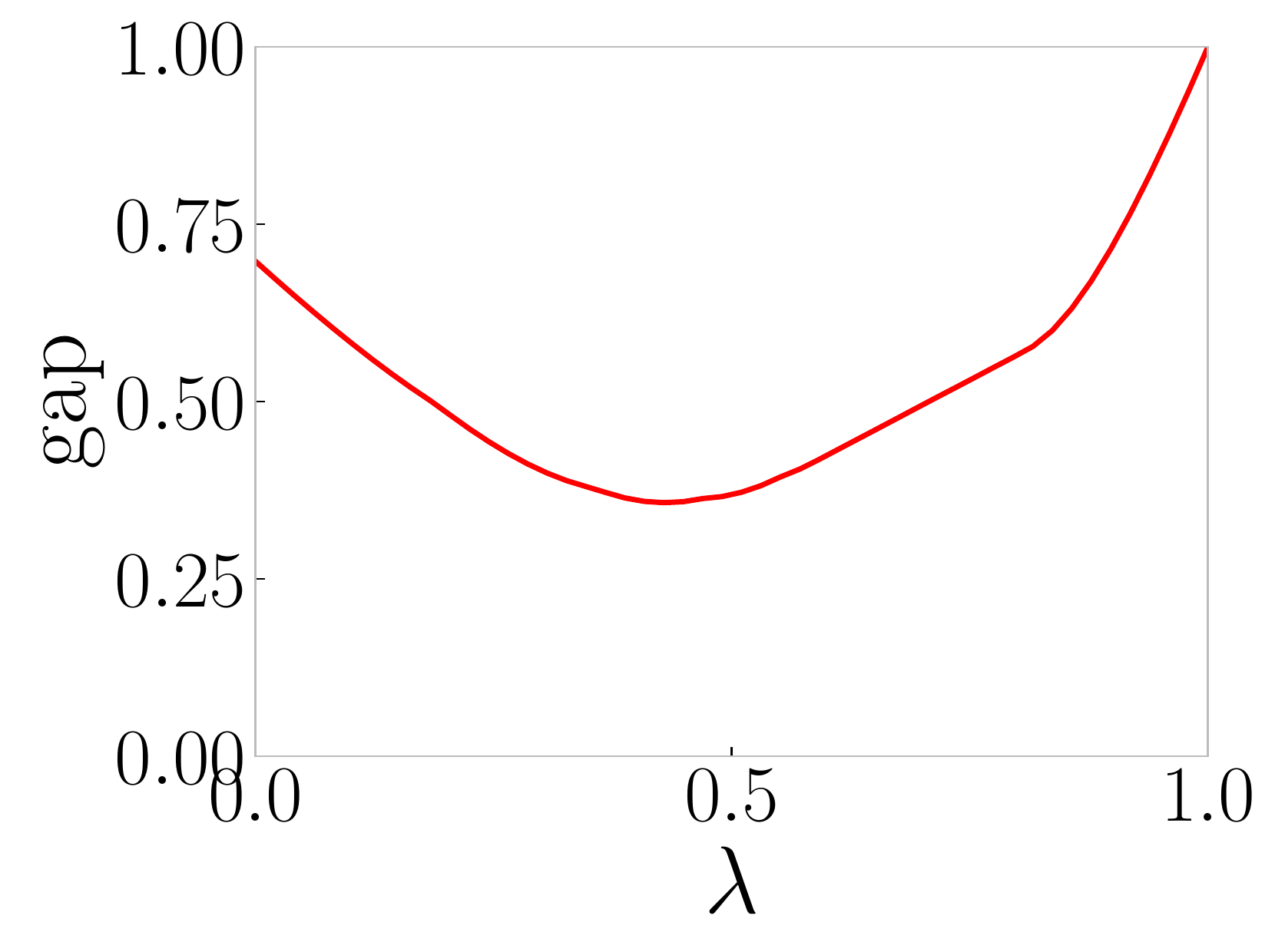}
    \caption{Adiabatic path connecting two decoupled Chern insulators ($\lambda=0$) to a product state ($\lambda=1$). The Hamiltonian is given in Eq.~\eqref{eq:Chern} and the couplings are depicted in Fig.~\ref{fig:chernstack}.}
    \label{fig:chernstackgap}
\end{figure}

\section{Further details on the KW maps}\label{app:KW}
\subsection{$\KW^G$ from the generalized cluster state}\label{app:Gclustestate}
We can explicitly construct the map $\KW^G$  as tensor network operator and demonstrate that it has the correct action. Denote $\ket{+}^G = \frac{1}{\sqrt{|G|}}\sum_{g\in G} \ket{g}$ and $\ket{1}^G = \ket{1^G}$ where $1^G$ is the identity element of $G$. Define
\begin{align}
    \KW^G_{EV} = \bra{+}^G_V U_{EV}^G \ket{1}^G_E
    \label{eq:KWGimplement}
\end{align}
where
\begin{align}
   U_{EV}^G=\prod_v\left [ \prod_{e\rightarrow v} CR^\dagger_{ve} \prod_{e\leftarrow v} CL^\dagger_{ve}  \right] = \prod_e CL^\dagger_{i_ee} CR^\dagger_{f_ee}.
   \label{eq:UG}
\end{align}

The unitary $U^G$ is the generalization of the cluster state entangler for $G$ degrees of freedom, which has both $G$ symmetry and $\Rep(G)$ $(d-1)$-form symmetry\footnote{Here, our construction differs slightly from Ref.~\onlinecite{Brell2015} by using $CL^\dagger$ and $CR^\dagger$ instead of $CL$ and $CR$. This is so that the $G$ 0-form symmetry acts as left multiplication instead of right multiplication, and is more compatible with the usual definition of a factor system where $\sigma$ acts from the left. See Appendix A-B of Ref.~\onlinecite{TJV2} for an example of the calculation with right multiplication.}. The Controlled-Not gates are generalized to controlled left and right multiplication operators
\begin{align}
\begin{split}
CL_{ve} \ket{g_1}_v\ket{g_2}_e &= \ket{g_1}_v\ket{g_1g_2}_e, \\ 
CR_{ve} \ket{g_1}_v\ket{g_2}_e &= \ket{g}_v\ket{g_2\bar g_1}_e
\end{split}
\label{eq:controlled-not}
\end{align}
An intuitive way to see the action of $U^G$ is to note that for each edge $e$, the unitary acts as
\begin{align}
    CL^\dagger_{i_ee}CR^\dagger_{f_ee}\ket{g_{i_e}}\ket{g_{e}}\ket{g_{f_e}} = \ket{g_{i_e}}\ket{\bar g_{i_e}g_{e} g_{f_e}} \ket{g_{f_e}}.
    \label{eq:clusterstateproperty}
\end{align}
Thus, after setting $g_e=1$ (which is implemented by contracting with $\ket{1}^G_E$), we are left with $\bar g_{i_e} g_{f_e}$, the domain-wall variable we map to in Eq.~\eqref{equ:KWonstatesG}.

Let us show that $\KW^G$ defined in Eq. ~\eqref{eq:KWGimplement} has the correct kernels. The following properties of the controlled left and right multiplication operators will be useful.

\begin{align}
CL_{ve}^\dagger (L^g_v \otimes L^g_e )CL_{ve} &= L^g_v \otimes \mathbbm1_e, & CR_{ve}^\dagger(L^g_v \otimes R^g_e)CR_{ve} &= L^g_v \otimes \mathbbm1_e,\label{eq:CL1}\\
CL_{ve}^\dagger(R^g_v \otimes  \mathbbm 1_e) CL_{ve}  &= R^g_v \otimes L^g_e, & CR_{ve}^\dagger(R^g_v \otimes  \mathbbm 1_e)CR_{ve} &= R^g_v \otimes R^g_e,\label{eq:CL2}\\
CL_{ve}(\bs Z^\mu_v \otimes \bs Z^\mu_e)CL_{ve}^\dagger &=\mathbbm 1_v \otimes \bs Z^\mu_e, & CR_{ve}(\bs Z^{\bar \mu}_v \otimes \bs Z^\mu_e)CR_{ve}^\dagger &=\mathbbm 1_v \otimes \bs Z^\mu_e . \label{eq:CL4}
\end{align}

First let us show that left multiplication is a symmetry, Eq.~\eqref{eq:G0form},
\begin{align}
\begin{split}
\KW^G  \times \prod_v L^g_v &=\bra{+}^G_V U^G \prod_v L^g_v \ket{1}^G_E \\
&= \bra{+}^G_V U^G \prod_v\left [ L^g_v \prod_{e\rightarrow v} R^g_e\prod_{e\leftarrow v} L^g_e \right] \ket{1}^G_E \\
&=\bra{+}^G_V \prod_v L^g_v U^G \ket{1}^G_E  = \bra{+}^G_V  U^G \ket{1}^G_E =\KW^G
\end{split}
\end{align}
where in the second line  we used the fact that  $\displaystyle \prod_v  \left [ \prod_{e\rightarrow v} R^g_v\prod_{e\leftarrow v} L^g_v \right]\ket{1}^G_e =\ket{1}^G_e$ since each edge is acted by $L^g_eR^g_e$ which leaves $1^G$ invariant, and on the third line we used \eqref{eq:CL1} to obtain $\displaystyle U^G  \left[L^g_v \prod_{e\rightarrow v} R^g_v\prod_{e\leftarrow v} L^g_v\right]  \left(U^G\right)^\dagger =  L^g_v$ and finally $\bra{+}^G_v L^g_v = \bra{+}^G_v$.

Next, to show the dual symmetry \eqref{eq:repGd-1form}, we notice that
\begin{align}
\begin{split}
\bs Z^\mu_e  \times  \KW^G &= \bra{+}^G_V \bs Z^\mu_e  U^G  \ket{1}^G_E \\
 &= \bra{+}^G_V   U^G  \bs Z^{\bar\mu}_{i_e}\bs Z^\mu_e \bs Z^{\mu}_{f_e} \ket{1}^G_E \\
&=\bra{+}^G_V   U^G  \bs Z^{\bar\mu}_{i_e}\bs Z^{\mu}_{f_e}  \ket{1}^G_E = \KW^G  \times   \bs Z^{\bar\mu}_{i_e}\bs Z^{\mu}_{f_e}
\label{eq:ZZmapG}
\end{split}
\end{align}
where on the second line we used \eqref{eq:CL4} to obtain $\displaystyle \left(U^G\right)^\dagger\bs  Z^\mu_e U^G = \bs Z^{\bar\mu}_{i_e} \bs  Z^\mu_e\bs Z^{\mu}_{f_e}$
and on the third line we used $ \bs Z^\mu_e\ket{1}^G_E= \bs {\mathbbm 1}_e\ket{1}^G_E$.  Then Eq.~\eqref{eq:repGd-1form} follows immediately by sending through instead a product of $\bs Z^\mu_e$ around a closed loop and noticing that for each vertex along this loop we have $\bra{+}^G_v \displaystyle \bs Z^{\bar \mu}_v \bs Z^\mu_v  = \bra{+}^G_v \bs{\mathbbm 1}_v$ since for each $\bra{g}$ in the sum, we have $\bs\rho^{\bar \mu}(g) \bs\rho^{\mu}(g) = \bs\rho^{ \mu}(\bar g) \bs\rho^{\mu}(g) = \bs\rho^\mu(1^G) = \bs{\mathbbm 1}$

Lastly, we figure out the result of ``gauging" $R^g_v$ Eq.~\eqref{eq:LgtoAg}
\begin{align}
\begin{split}
\KW^G  \times  R^g_v &=\bra{+}^G_V U^G R^g_v \ket{1}^G_E\\
&= \bra{+}^G_V  R^g_v \prod_{e\rightarrow v} R^g_e\prod_{e\leftarrow v} L^g_e U^G  \ket{1}^G_E\\
&= \bra{+}^G_V  \prod_{e\rightarrow v} R^g_e\prod_{e\leftarrow v}  L^g_e U^G  \ket{1}^G_E = A^g_v  \times  \KW^G 
\end{split}
\end{align}
where on the second line we used \eqref{eq:CL2} to obtain $
   \displaystyle U^GL^g_v\left(U^G\right)^\dagger = R^g_v \prod_{e\rightarrow v} R^g_e\prod_{e\leftarrow v} L^g_e$
and on the third line we used $\bra{+}^G_V R^g_v = \bra{+}^G_V$.

\subsection{More on factor systems}\label{app:factor}
Here, we give a careful derivation of the properties regarding factor systems and how they determine the group extensions. 

Given a group $G$ and a normal subgroup $N \triangleleft G$, one has an exact sequence
\begin{equation}
 1 \rightarrow N  \xrightarrow[]{\iota} G  \xrightarrow[]{\pi} Q \rightarrow 1
 \label{eq:extension}
\end{equation}
That is, there exists an injective map $\iota: N \rightarrow G$ and a surjective map $\pi:G \rightarrow Q$ such that $\pi \circ \iota = 1$. Next, we pick a lift $s:Q\rightarrow G$ such that $\pi \circ s$ is the identity map in $Q$. Note that $s$ is not a group homomorphism. The map $s$ allows us to define the two pieces of data in the factor system

\begin{enumerate}
    \item The map $\sigma:Q\rightarrow \text{Aut}(N)$ can be defined by $$\iota(\sigma^q[n]) = s(q) \iota(n) \overline{s(q)}.$$ That is for each $q \in Q$, $\sigma^q$ defines an automorphism on $N$ given by conjugation with $s(q)$.
    \item The cocycle $\omega:Q^2\rightarrow N$ can be defined by
    $$ \iota(\omega(q_1,q_2)) = s(q_1) s(q_2)\overline{s(q_1q_2)},$$
    which captures the failure of $s$ to be a group homomorphism.
\end{enumerate}

Note that because $s(1^Q) =1^G$ we automatically have the following properties, assumed in the main text:
\begin{enumerate}
\item $\sigma^1$ acts as the identity automorphism.
\item $\omega$ is counital (``normalized") i.e. $\omega(1^Q,q) = \omega(q,1^Q)=1^N$.
\end{enumerate}
We remark that $\sigma$ might fail to be a group homomorphism if $N$ is non-Abelian
\begin{align}
    \sigma^{q_1} \circ \sigma^{q_2} = c^{\omega(q_1,q_2)} \circ \sigma^{q_1q_2}
\end{align}
where $c^{n_1}$ is an inner automorphism which acts as conjugation: $c^{n_1}[n_2]=n_1n_2\bar{n}_1$.

Let us define the elements of the extended group $g\in G$ as $g= (n,q) \equiv \iota(n) s(q)$. From this, it follows that the group multiplication is given by
\begin{align}
    (n_1,q_1) \times (n_2,q_2) = (n_1\aut{q_1}{n_2}\omega(q_1,q_2),q_1q_2)
    \label{eq:multiplicationlawapp}
\end{align}
Furthermore, associativity of group multiplication requires $\omega$ to satisfy the following cocycle condition
\begin{equation}
    \aut{q_1}{\omega(q_2,q_3)} \omega(q_1,q_2q_3)  = \omega(q_1,q_2) \omega(q_1q_2,q_3).
    \label{equ:cocycleconditionapp}
\end{equation}

\subsection{Properties of $\KW^{N\triangleleft G}$}\label{app:KWN}

We would like to show that $\KW^{N \triangleleft G}_{EV}$ (defined as the unique map that gauges $N$ and leaves $Q$ invariant) can be expressed as
\begin{align}
   \KW^{N \triangleleft G}_{EV} =\mathcal U^{N \triangleleft G}_{EV} \times \KW^N_{EV}
\end{align}
Equivalently, it suffices to define $\KW^{N \triangleleft G}$ as above, and show that $\KW^{N \triangleleft G}$ has the correct kernels.  First, it is helpful to define the unitary
\begin{align}
   U^{N \triangleleft G}_{EV} = \mathcal U^{N \triangleleft G}_{EV} U^N_{EV} =\Sigma_{EV}^{-1} \Omega_{VEV} U^N_{EV}
\end{align}
so that the map can be expressed as
\begin{align}
   \KW^{N \triangleleft G}_{EV} = \bra{+}^N_V U_{EV}^{N \triangleleft G} \ket{1}^N_E.
\end{align}

We note that similarly to Eq.~\eqref{eq:clusterstateproperty}, an intuitive way to see why this map works is to note that, $U^{N \triangleleft G} $ acts on each edge as
\begin{align}
U^{N \triangleleft G}_{EV} \ket{g_{i_e}}\ket{n_e}\ket{g_{f_e}}=\ket{g_{i_e}}\ket{\bar \omega(\bar q_{i_e},q_{i_e}) \aut{\bar q_{i_e}}{\bar n_{i_e}n_e n_{f_e}}\omega(\bar q_{i_e},q_{f_e})}\ket{g_{f_e}}
\end{align}
Therefore, after setting $n_e=1$ (which is implemented by contracting with $\ket{1}^N_E$), we are left with $t(\bar g_{i_e}g_{f_e})$, the domain wall variable we map to in Eq.~\eqref{eq:tgigfbar}.

For calculation purposes, it is also useful to have
\begin{align}
(U^{N \triangleleft G}_{EV})^\dagger \ket{g_{i_e}}\ket{n_e}\ket{g_{f_e}}   &=\ket{g_{i_e}}\ket{n_{i_e} \aut{q_{i_e}}{n_e}\omega(q_{i_e},\bar q_{i_e}q_{f_e}) \bar n_{f_e}}\ket{g_{f_e}}   
\end{align}
as well as the action of left and right multiplications in the basis $\ket{n,q} \in \CC[N]\otimes \CC[Q]$.
\begin{align}
  L^g_v \ket{n_v,q_v} &=   \ket{n\aut{q}{n_v}\omega(q,q_v),qq_v} & R^g_v \ket{n_v,q_v} &=   \ket{n_e\bar \omega(q_v\bar q,q)\aut{q_v\bar q}{\bar n},q_v\bar q}
\end{align}

From the above, one finds the following identities
\begin{align}
   (U^{N \triangleleft G})^\dagger \left( \prod_v L^g_v \right)  U^{N \triangleleft G} &= \prod_v L^g_v \prod_e L^n_e R^n_e \Sigma^q_e \label{eq:conjL_N}\\
 (U^{N \triangleleft G})^\dagger \tilde{\bs Z}^\nu_e  U^{N \triangleleft G}& = {\bs Z}^{\bar\nu}_{i_e} {\bs Z}^{\nu}_{e} {\bs Z}^{\nu}_{f_e} \label{eq:conjZtilde_N}
\end{align}
where $\Sigma^q = \sum_n \ket{\sigma^q[n]}\bra{n}$

Now let us check that $\KW^{N \triangleleft G}$ has the correct kernels, to show Eq. \eqref{eq:Q0form},
\begin{align}
\KW^{N \triangleleft G}  \times \prod_v L^g_v &=\bra{+}^N_V U^{N \triangleleft G} \prod_v L^g_v \ket{1}^N_E \nonumber \\
&=\bra{+}^N_VU^{N \triangleleft G} \left[\prod_v L^g_v \prod_e L^n_e R^n_e \Sigma^q_e\right]  \ket{1}^N_E \\
&=\bra{+}^N_V \prod_v L^g_v U^{N \triangleleft G}   \ket{1}^N_E = \bra{+}^N_V \prod_v L^q_v U^{N \triangleleft G}   \ket{1}^N_E =\prod_v L^q_v \times  \KW^{N \triangleleft G}
\end{align}
where on the second line we used the fact that $L^n_e R^n_e \Sigma^q_e\ket{1^N}_e =\ket{n\sigma^q[1^N]\bar n}_e = \ket{1^N}_e$, on the third line we used Eq.~\eqref{eq:conjL_N} and finally $\bra{+}^N_V \prod_v L^g_v = \bra{+}^N_V \prod_v L^q_v$.

Next, to show the dual $\Rep(N)$ $(d-1)$-form symmetry, consider
\begin{align}
     \tilde{\bs  Z}^{\nu}_e \times \KW^{N\triangleleft G} &= \bra{+}^N_V \tilde{\bs  Z}^{\nu}_e   U^{N \triangleleft G} \ket{1}^N_E\\
&=\mathcal U \bra{+}^N_V  U^{N \triangleleft G}    {\bs Z}^{\bar\nu}_{i_e} {\bs Z}^{\nu}_{e} {\bs Z}^{\nu}_{f_e} \ket{1}^N_E\\
&=\mathcal U \bra{+}^N_V  U^{N \triangleleft G} {\bs Z}^{\bar\nu}_{i_e}  {\bs Z}^{\nu}_{f_e} \ket{1}^N_E=\KW^{N\triangleleft G} \times   {\bs Z}^{\bar\nu}_{i_e}  {\bs Z}^{\nu}_{f_e}
\end{align}
where on the second line we used  Eq.~\eqref{eq:conjZtilde_N}, and on the third line we used $\bs Z^\nu_e\ket{1}^N_E = \mathbbm 1_e\ket{1}^N_E$. Thus, by taking a product over all edges in a closed loop, the  ${\bs Z}^{\bar\nu}_{i_e}  {\bs Z}^{\nu}_{f_e}$ terms pairwise cancel at each vertex, proving Eq. \eqref{eq:repNd-1form}.

\section{Nil-2 Quantum Doubles}\label{app:nil2}

\subsection{Proof of preparation}\label{app:nil2prepproof}
Let us confirm that the protocol in Eq.~\eqref{eq:nil2} indeed prepares $\mathcal D(G)$. An important observation is that the intermediate state after gauging $N$ is exactly the toric code state. Therefore,
\begin{align}
    \ket{\mathcal D(N)}_E =\KWd_{EP}^N\ket{+}^N_P = \KW_{EV}^N\ket{+}^N_V
\end{align}
where instead of measuring the plaquette operators, we instead prepared the toric code by measuring the vertex operators. Therefore,
\begin{align}
\KW^Q_{EV}\Omega_{VEV} \KWd^N_{EP}\ket{+}^Q_V \ket{+}^N_P &=  \KW^Q_{EV}\Omega_{VEV}\KW_{EV}^N\ket{+}^Q_V \ket{+}^N_V \\
&= \KW^Q_{EV}\KW_{EV}^{N\triangleleft G} \ket{+}^Q_V \ket{+}^N_V \\
&= \KW^G_{EV} \ket{+}^G_V  = \ket{\mathcal D(G)}_E
\end{align}
where we used the definition of $\KW^{N \triangleleft G}$ in Eq.~\eqref{eq:KWNGdef} for the case of a central extension (trivial $\sigma$ and therefore $\Sigma_{EV} = \mathbbm 1$) and the two-step gauging property in Eq.~\eqref{eq:2step}.

\subsection{One-shot preparation of nil-2 quantum double by gauging a decorated domain wall SPT state}\label{app:DDW}
An alternative way to understand the protocol in Eq.~\eqref{eq:nil2} is to treat the state preparation as gauging an SPT state protected by $N \times Q$. First, note that if $\Omega_{VEV}$ were absent, then we are simply applying KW to gauge the product state with symmetry $N \times Q$, which will give an $N\times Q$ toric code. However, as pointed out in the main text, by using $\Omega_{VEV}$ we are turning the $N$-toric code phase into a non-trivial SET phase protected by $Q$. In fact, we can further push $\Omega_{VEV}$ through the $\KWd^N$ by using the fact that it performs a right multiplication by $\omega(q_{i_e},\bar q_{i_e}q_{f_e})$, thus pushing it through and using the fact that $\KWd^{N}$ turns ``$X$" (in this case, right multiplication) into ``$ZZ$",
\begin{align}
    \Omega_{VEV}  \KWd^N_{EP} &= \KWd^N_{EP} \Omega_{VVPP},\\
    \Omega_{VVPP} \ket{q_{i_e}, q_{f_e},n_{i_{\check{e}}},n_{f_{\check{e}}} } &= \chi^{\bar \omega(q_{i_e},\bar q_{i_e}q_{f_e})}(\bar n_{i_{\check{e}}}  n_{f_{\check{e}}} )\ket{q_{i_e}, q_{f_e},n_{i_{\check{e}}},n_{f_{\check{e}}} }
\end{align}
where $i_{\check{e}}$ and  $f_{\check{e}}$ denotes the plaquettes with dotted lines pointing into and out of the edge $e$. This is exactly the decorated domain wall wavefunction\cite{decorateddomainwalls}. For each edge $e$, a ``1D SPT state" given by a 2-cocycle $\omega$ whose wavefunction lives on the vertices and is given by $\omega(q_{i_e}\bar q_{f_e},q_{f_e})$ is present whenever there is an $N$ domain wall, which is detected by the combination $n_{i_{\check{e}}}\bar n_{f_{\check{e}}}$. To conclude, our state has been recasted into the form
\begin{align}
     \ket{\mathcal D(G)}_E &=\bra{+}^Q_{V} CX_{VE} \ket{1}^Q_E \times   \bra{+}^N_P CZ_{PE}\ket{+}^N_E    \times  \ket{\SPT} 
\end{align}
where the SPT protected by $N\times Q$ is given by
\begin{align}
 \ket{\SPT}= \Omega_{VVPP} \ket{+}^Q_V \ket{+}^N_P.
\end{align}
Hence in this viewpoint, we are preparing a twisted quantum double $\mathcal D^\alpha(N\times Q)$ realizing the same phase of matter as $\mathcal D(G)$\footnote{Mathematically, this follows from the fact that $\text{Vec}(G)$ and $\text{Vec}^\alpha(N\times Q)$ are Morita equivalent as fusion categories. Therefore, the topological order in the bulk are related by a circuit \cite{Lootens22}. Nevertheless, our preparation scheme here is designed such that we \textit{exactly} prepare $\mathcal D(G)$ and no extra circuit is required.}. The 3-cocycle $\alpha$ corresponds to a class $[\alpha] \in H^3(N \times Q,U(1))$, and can be related to $\omega$ using the $H^2(Q,H^1(N,U(1)))$ part of the Kunneth formula). The two classes are related via\cite{HuWan20,Tachikawa_2020}
\begin{align}
    \alpha = \rho \cup \omega
\end{align}
where $\rho$ is the generating class of $H^1(N,U(1))$.

\section{One shot preparation of 2-group gauge theory with Abelian 0-form symmetry}\label{app:2group}

A 2-group\footnote{not to be confused with $p$-group used in group theory where $|G|$ is some power of $p$.} $\mathbb G$ gauge theory can be specified by a $0$-form symmetry $Q$, a $1$-form symmetry $N$, an automorphism $\sigma: Q \rightarrow \text{Aut}(N)$ and a certain 3-cocycle $\omega\in H^3_\sigma(Q,N)$ called the Postnikov class. The natural generalization of a nil-2 group in this setting is to assume that $Q$ is Abelian, and that the automorphism $\sigma$ is trivial. Similar to the 2D construction, we start with a 3D triangulation with branching structure and place $\CC[Q]$ on vertices $V$ and edges $E$, and $\CC[N]$ on plaquettes $P$ and tetrahedra $T$.
\begin{align}
    \ket{\mathbb G}_{EP} &=\KW^Q_{EV} \Omega_{VPV} \KWd^N_{PT} \ket{+}^Q_V \ket{+}^N_T \nonumber\\
    &=\bra{+}_{VT} CX^Q_{VE}   \Omega_{VPV}    CZ^N_{TP} \ket{+}^Q_{V}\ket{1}^Q_E \ket{+}^N_P\ket{+}^N_T
\end{align}
To define $\Omega_{VPV}$ the vertices in each plaquette $p$ can be ordered using the branching structure as $v_{p0},v_{p1},v_{p2}$. Then,
\begin{align}
      \Omega_{VPV} &=\prod_p \ket{q_{v_{p0}},q_{v_{p1}} q_{v_{p1}}, n_p \bar \omega(q_{v_{p0}},\bar q_{v_{p0}}q_{v_{p1}},\bar q_{v_{p1}}q_{v_{p2}} )}\bra{q_{v_{p0}},q_{v_{p1}} q_{v_{p1}},n_p} 
\end{align}

Similarly, this construction can be thought of as gauging a $N\times Q$ SPT where both symmetries are $0$-form with 4-cocycle $\alpha = \rho \cup \omega$ where $\rho$ is the generating class of $H^1(N,U(1))$. After commuting $\Omega_{VPV}$ through $\KWd^N_{PT}$ we obtain
\begin{align}
\ket{\mathbb G}_E &=\KW^Q_{EV} \KWd^N_{PT}  \Omega_{VVVTT} \ket{+}^Q_V \ket{+}^N_T
\end{align}
where $\Omega_{VVVTT}$ is the entangler for the SPT by decorating a 2D $Q$- ``SPT" on $N$ domain walls
\begin{align}
      \Omega_{VVVTT}\ket{q_{v_{p0}},q_{v_{p1}} q_{v_{p2}},n_{{i_p}},n_{{f_e}}}  &=
      \chi^{\bar \omega(q_{v_{p0}},\bar q_{v_{p0}}q_{v_{p1}},\bar q_{v_{p1}}q_{v_{p2}} )}(\bar n_{t_{i_p}}n_{t_{f_p}})
      \ket{q_{v_{p0}},q_{v_{p2}} q_{v_{p1}},n_{{i_p}},n_{{f_e}}} 
\end{align}
Here, ${i_p}$ and ${f_p}$ are tetrahedra that point into and out of the plaquette $p$.

\end{document}